\pdfoutput=1
\documentclass[nonacm]{acmart}

\usepackage{xspace} 
\usepackage{listings}
\usepackage{enumitem}

\usepackage[ruled]{algorithm2e}
\RestyleAlgo{ruled}

\SetCommentSty{mycommfont}

\usepackage{amsmath}
\usepackage{amsthm}
\newtheorem*{remark}{Remark}

\usepackage{subfig}
\usepackage{graphicx}

\usepackage{tabularx}
\usepackage{array}
\usepackage{ragged2e}
\usepackage{booktabs}

\newcommand{\ours}{\textsc{Ours}\xspace}

\newcommand{\ucb}{\textsc{UCB}\xspace}
\newcommand{\scan}{\textsc{Scan}\xspace}
\newcommand{\sample}{\textsc{UniformSample}\xspace}
\newcommand{\explore}{\textsc{ExplorationOnly}\xspace}
\newcommand{\scanbest}{\textsc{ScanBest}\xspace}
\newcommand{\scanworst}{\textsc{ScanWorst}\xspace}
\newcommand{\sortedscan}{\textsc{SortedScan}\xspace}

\newcommand{\tth}{\textsuperscript{th}\xspace}

\newcounter{myprocctr}
\newenvironment{proc}{
   \refstepcounter{myprocctr}
   \textsc{Procedure \themyprocctr.}
}{\par}  
\numberwithin{myprocctr}{section}

\usepackage{alltt}

\AtBeginDocument{%
  }

\setcopyright{cc}
\acmJournal{PACMMOD}
\acmYear{2025} \acmVolume{3} \acmNumber{3 (SIGMOD)}
\acmArticle{129} \acmMonth{6} 
\acmDOI{10.1145/3725266}

\begin{document}

\title{Approximating Opaque Top-k Queries}

\author{Jiwon Chang}
\email{jchang38@ur.rochester.edu}
\affiliation{%
  \institution{University of Rochester}
  \city{Rochester}
  \state{NY}
  \country{USA}
}
\author{Fatemeh Nargesian}
\email{fnargesian@rochester.edu}
\affiliation{%
  \institution{University of Rochester}
  \city{Rochester}
  \state{NY}
  \country{USA}
}

\renewcommand{\shortauthors}{Chang and Nargesian}

\begin{abstract}
    Combining query answering and data science workloads has become prevalent. 
    An important class of such workloads is top-k queries with a scoring function implemented as an opaque UDF---a black box whose internal structure and scores on the search domain are unavailable. 
    Some typical examples include costly calls to fuzzy classification and regression models. 
    The models may also be changed in an ad-hoc manner.
    Since the algorithm does not know the scoring function's behavior on the input data, opaque top-k queries become expensive to evaluate exactly or speed up by indexing. 
    Hence, we propose an approximation algorithm for opaque top-k query answering. 
    Our proposed solution is a task-independent hierarchical index and a novel bandit algorithm. 
    The index clusters elements by some cheap vector representation then builds a tree of the clusters. 
    Our bandit is a diminishing returns submodular epsilon-greedy bandit algorithm that maximizes the sum of the solution set's scores. 
    Our bandit models the distribution of scores in each arm using a histogram, then targets arms with fat tails. 
    We prove that our bandit algorithm approaches a constant factor of the optimal algorithm. 
    We evaluate our standalone library on large synthetic, image, and tabular datasets over a variety of scoring functions. 
    Our method accelerates the time required to achieve nearly optimal scores by up to an order of magnitude compared to exhaustive scan while consistently outperforming baseline sampling algorithms. 
\end{abstract}



\maketitle

\section{Introduction}
\label{sec:intro} 

Intermixing opaque user defined functions (UDF) with query answering has become an important workload for performing data science on unstructured text, images, and tabular data~\cite{foufoulas2023efficient}. 
An opaque UDF is a black-box function whose semantics and internal structure nor its scores on the database are available a priori~\cite{he2020method}. 
They are invoked in stand-alone code bases~\cite{spiegelberg2021tuplex,he2020method,he2024optimizing} or intermixed with SQL queries~\cite{he2020method,anderson2016input,he2024optimizing,sikdar2020monsoon}. 
Opaque UDFs are typically written in conventional programming languages and they increasingly involve calls to machine learning models~\cite{foufoulas2023efficient,spiegelberg2021tuplex,dai2024uqe,grulich2021babelfish}, as emerging academic and industry systems integrate ad-hoc model training and inference into databases with declarative interfaces.  
For example, users may train models on-the-fly to predict missing values~\cite{hasani2019approxml,jin2024aidb,zhao2024neurdb} or prompt large language models (LLMs)~\cite{dai2024uqe,anderson2024design,liu2024optimizing,patel2024lotus,noauthor_ai_query_nodate,noauthor_llm_nodate} via an SQL extension.

\subsection{Opaque Top-\textit{k} Queries} 

Opaque  top-$k$ queries rank elements in a database by an opaque scoring function, and return the $k$ highest scoring elements.
\begin{alltt}
\begin{flushleft}
    \textbf{SELECT} * \textbf{FROM} vehicleListings\\
    \textbf{ORDER BY} ValuationUDF(listing) \textbf{LIMIT} 250;
\end{flushleft}
\end{alltt}
For example, an analyst may want to retrieve 250 current vehicle listings that have the highest predicted valuations by applying a price prediction ML model on a database of listings.

Traditional top-$k$ query answering literature generally assumes certain properties on scoring functions~\cite{ilyas2008survey,fagin2001optimal}.  
For example, the family of Threshold Algorithm (TA) variants typically assume that the scoring function is a monotone aggregate~\cite{fagin1996combining,fagin2001optimal,ilyas2008survey,bruno2007threshold,shmueli2009best}. Another class of scoring functions is similarity functions, such as inner product of vector representations~\cite{pan2024survey}. 
These function classes are prevalent yet suffer from limited expressivity, potentially losing complex user intent.

Arbitrary UDFs, however, pose fundamental challenges to top-$k$ query evaluation. First, they may be difficult to analyze, as they involve imperative programming languages and ML models. 
Existing methods increase programmer burden~\cite{spiegelberg2021tuplex} or limit the scope~\cite{hagedorn2021putting} in order to expose the internals of the UDF to the compiler. 
Second, they may be expensive to run---foundation models are the extreme case~\cite{dai2024uqe,liu2024optimizing,noauthor_ai_query_nodate,noauthor_llm_nodate}. 
Hence, an exhaustive scan is not attractive. 
Third, the scores given by the function may change over time, if the scoring function is a continually learning model~\cite{de2021continual} or if the score depends on some external factor, like a user-specified label or prompt. As such, it is hard to build a sorted index over the elements' scores or apply the existing approximate nearest neighbor search techniques~\cite{WangXY021}. Concretely, consider the following scenario. 

\begin{example}
    Analyst Alice is developing an in-house model to predict the market prices of used vehicles. The model is updated frequently to adapt to changing market preferences. She is interested in identifying which listings have the highest predicted valuations, as this would help her company target its most valuable customers.
    To achieve this, she trains a gradient-boosted decision tree model for tabular regression~\cite{chen2016xgboost}. This UDF is an opaque scoring function: it invokes a model and incurs significant inference latency, and it is constantly evolving. 
    Alice then issues a top-$k$ query using this UDF to find the 250 highest-valued listings from a dataset of hundreds of thousands. As she continues to refine the model, she will repeatedly issue new top-$k$ queries. Since she wants to interactively analyze the dataset and the models, she is willing to sacrifice some quality for reduced latency.
    \label{ex:running}
\end{example}

Traditional methods are unsatisfactory for the query given in Example~\ref{ex:running}. 
Conventional data engines such as Spark or PostgreSQL treat UDFs as a black box, and will execute the UDF exhaustively on all data points~\cite{he2020method,cao2004efficient}. 
This is not viable when the search domain is very large, each invocation of the model is costly, or the user does not need the exact top-$k$ solution. 
Taking a sample of the search domain is also viable, but degrades in quality when only a small subset of the elements have high scores. 
Uniform sampling, in the worst case, requires a sample size of $n (1-\varepsilon)$, where $n$ is the number of data points and $\varepsilon$ is approximation tolerance. 
Specialized algorithms for limited classes of scoring functions and indexes are not applicable, either.

A recent body of work studied how to statistically optimize opaque UDF queries. However, existing methods are limited to other types of queries, such as feature engineering~\cite{anderson2016input}, selection~\cite{he2020method,dai2024uqe}, aggregation~\cite{dai2024uqe}, and queries with partially obscured predicates~\cite{sikdar2020monsoon}. 
To the best of our knowledge, no prior work studied statistical optimization for opaque top-$k$ queries. 

A method for approximate opaque top-$k$ query execution has to balance conflicting objectives. 
It should be applicable across a wide variety of opaque scoring functions. 
It should not incur significant execution-time overhead in the worst-case scenario. 
It should be amenable to batched execution to maximize lower-level optimization opportunities. 
Finally, its performance-over-time curve should be close to optimal, and consistently better than random sampling.

\begin{figure}
    \centering
    \subfloat[Hierarchical cluster with \\histogram sketches.]
    {
        \includegraphics[width=0.47\columnwidth]{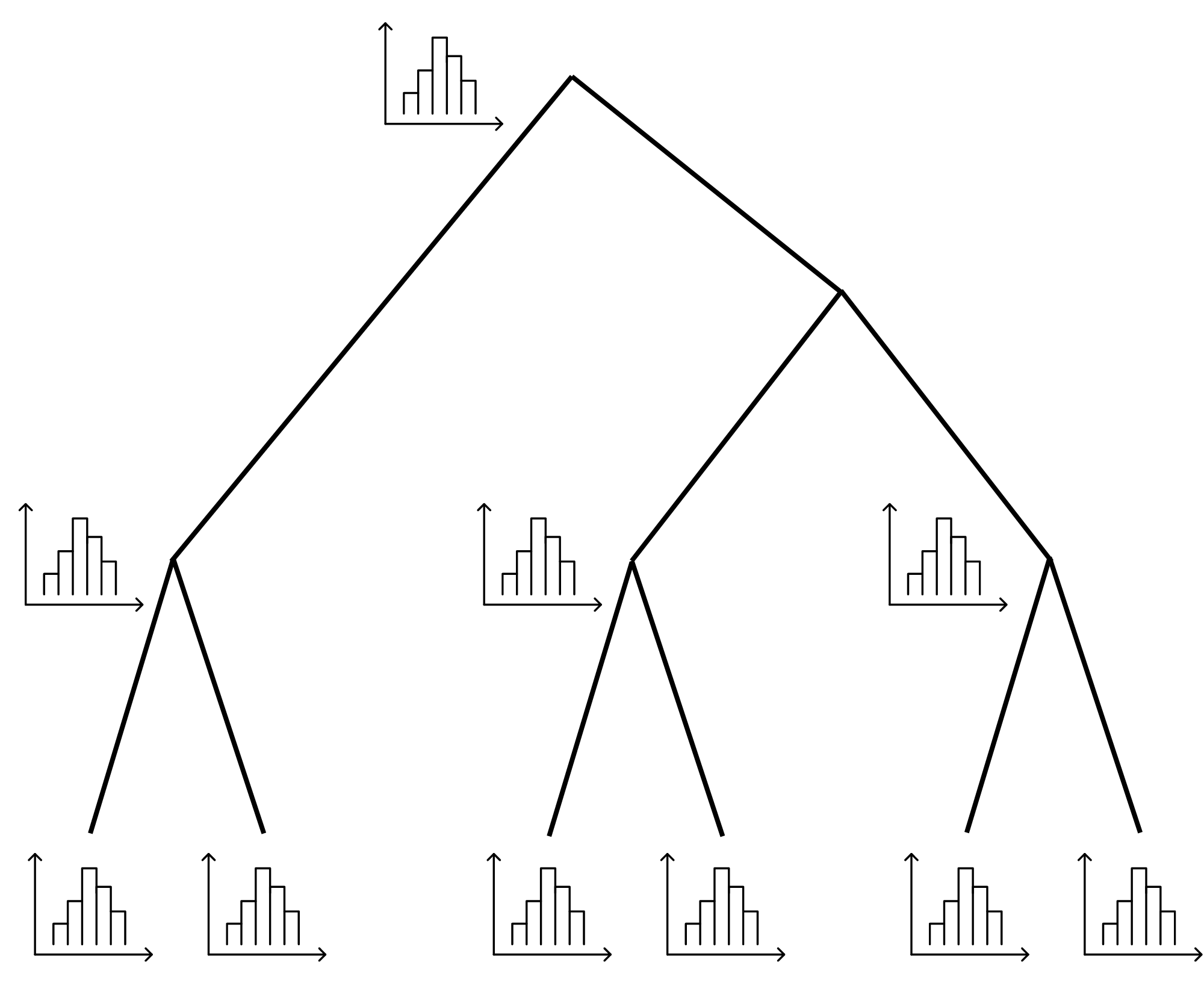}
        \label{fig:index-visualization}
    }
    \hfill
    \subfloat[Each iteration of our \\bandit algorithm.]
    {
        \includegraphics[width=0.48\columnwidth]{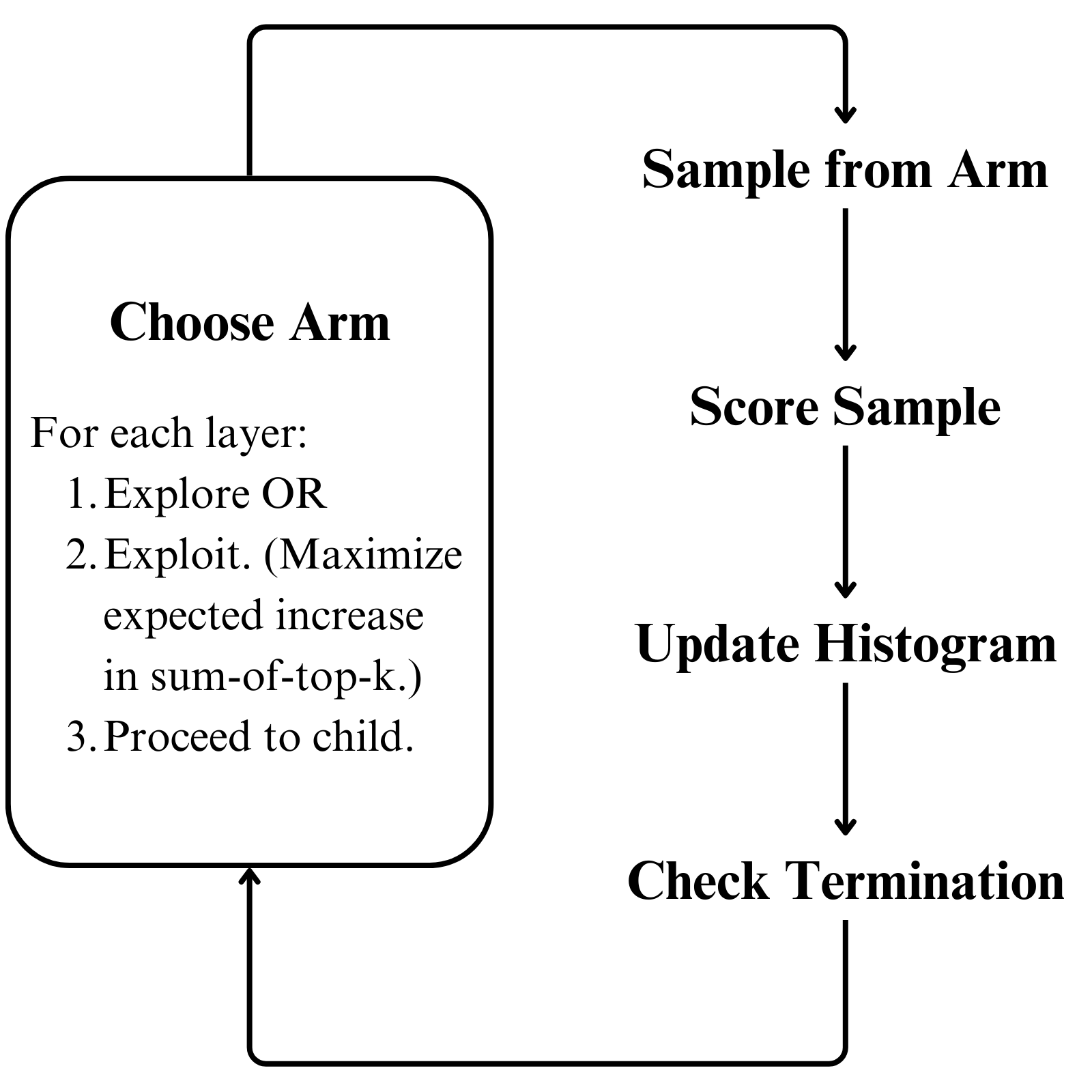}
        \label{fig:bandit-visualization}
    } 
    \caption{
    A high-level overview of our proposed solution 
    during indexing and query execution.  
    }
    \Description{Left: A non-balanced binary tree where each node, including leaves, have a histogram sketch attached. Right: The logical loop of one iteration of our bandit algorithm, which repeates: choose arm, sample from arm, score sample, update histogram, and check termination.}
    \label{fig:overview-visualization}
\end{figure}

\subsection{Solution Sketch} \label{sec:solution}

Our proposed solution, as illustrated in Figure~\ref{fig:overview-visualization}, consists of a tree index and a novel bandit algorithm. The tree index hierarchically clusters data at preprocessing time before any opaque top-$k$ queries. 

We adopt the VOODOO index from He et al.~\cite{he2020method}. This index vectorizes each element using a cheap heuristic, applies $k$-means clustering to the vectors, and then performs agglomerative clustering on the cluster centroids.

Our main technical contribution is the query execution algorithm. 
We abstract the problem of sampling over a collection of heterogeneous clusters for top-$k$ queries as a stochastic diminishing returns (DR) submodular bandit problem~\cite{soma2014optimal,asadpour2016maximizing}. 
In this problem setting, there are three classes of algorithms: offline, non-adaptive, and adaptive. 
The optimal offline algorithm is the theoretical best-case scan when the insertion order of the data is ideal. 
Adaptive algorithms change their behavior based on the realizations of random samples; non-adaptive algorithms do not. 
Prior work established gaps between the three classes~\cite{hellerstein2015discrete,asadpour2016maximizing}. 
We design an algorithm that performs close to the optimal adaptive, as illustrated in Figure~\ref{fig:rate-illustration}. 

\begin{figure}
    \centering
    \includegraphics[width=0.9\linewidth]{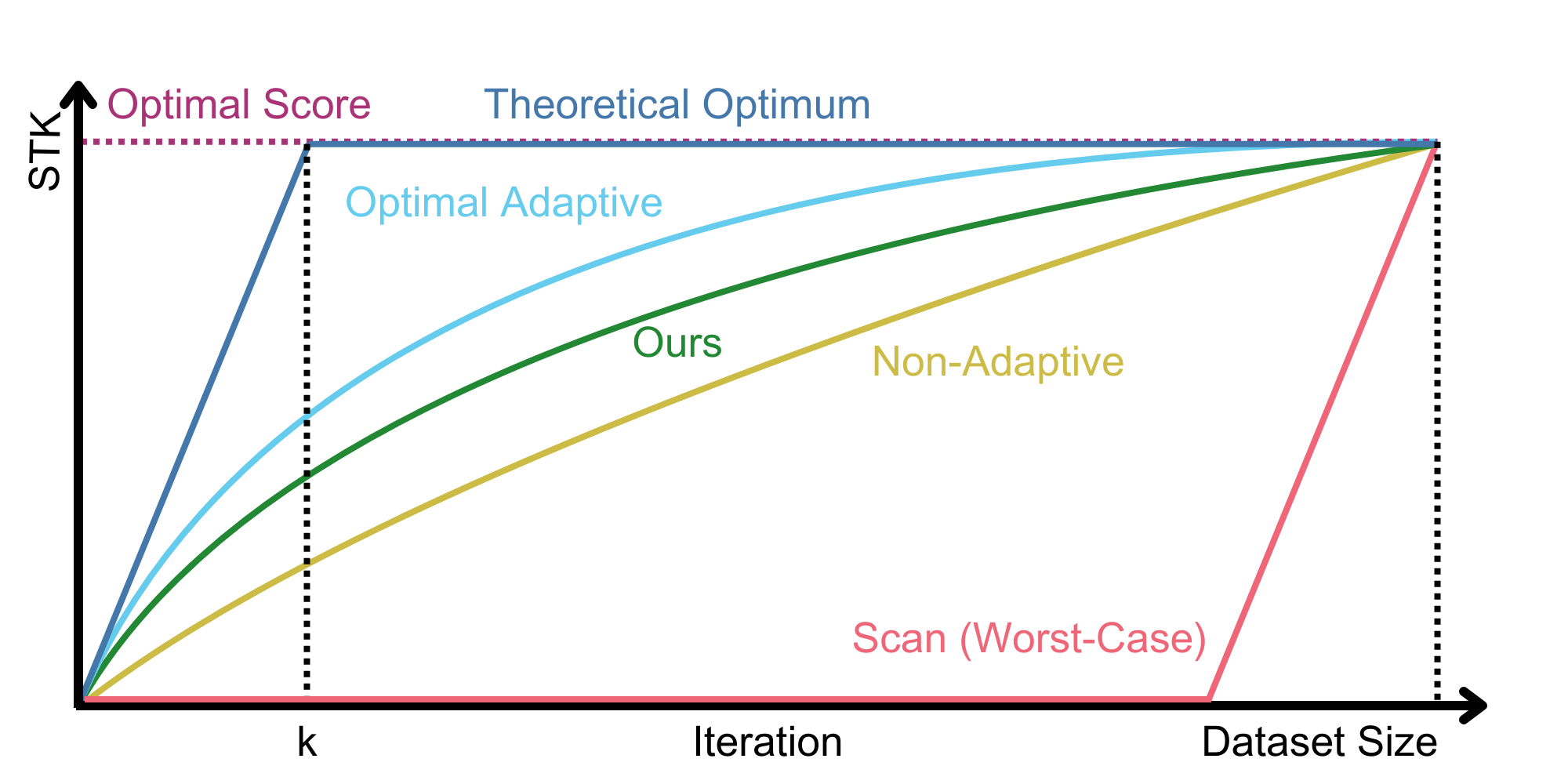}
    \caption{Visualizing the relative performance of various algorithms. $x$-axis is the number of iterations (i.e. scoring function calls) and the $y$-axis indicates the quality of results in terms of the sum of top-$k$ scores (STK).}
    \Description{The optimal score in terms of STK is some unknown constant. If the elements are ordered in the best or the worst case, \scan achieves the optimal score after just $k$ iterations, or as late as after $n - k$ iterations. The optimal adaptive algorithm is closest to the theoretical optimum, followed by \ours and non-adaptive sampling algorithms.}
    \label{fig:rate-illustration}
\end{figure}

Our main contributions are summarized below.
\begin{enumerate}
    \item We study the problem of approximating opaque top-$k$ queries. We formalize query answering over the index structure as a DR-submodular bandit problem. 
    \item We propose a histogram-based $\varepsilon$-greedy bandit algorithm.  
    \item We prove approximation guarantees for the objective of maximizing the total utility of the solution, given by its sum of scores. When the scoring domain is non-negative, discrete and finite, our algorithm's expected utility is lower bounded by $(1-\exp(-1-1/2T)) \text{OPT} - O(T^{2/3})$ at iteration limit $T$, which approaches 63\% of the optimal adaptive algorithm. 
    \item When the scoring distribution is an arbitrary non-negative interval, we describe a histogram maintenance strategy that improves flexibility and accuracy, and fallback strategies that improve practical performance.
    \item We implement our solution as a standalone library and perform extensive experiments on large synthetic, image, and tabular datasets over a variety of scoring functions. 
\end{enumerate}

The remainder of this paper is organized as follows. 
In \S~\ref{sec:problemdef}, we formally define the concrete and abstract problems. 
\S~\ref{sec:algorithms} presents our bandit algorithm. 
In \S~\ref{sec:analysis}, we establish the theoretical guarantees. 
\S~\ref{sec:experiments} reports the experimental results. 
We discuss related work in \S~\ref{sec:related}. Finally, we conclude with a discussion in \S~\ref{sec:discussion} and~\ref{sec:conclusion}.

\section{Problem Definition}\label{sec:problemdef}

This section defines two related problems: the first is the concrete problem of answering opaque top-$k$ queries, and the second is an abstract bandit problem that assumes an index. 

Refer to Table~\ref{tbl:notations} for the precise definition of notations.

\begin{table}[htb] 
    \centering
    \setlength{\extrarowheight}{2pt} 
    \begin{tabularx}{\columnwidth}{|>{\centering\arraybackslash}p{0.15\columnwidth}|
                                    >{\raggedright\arraybackslash}X|}
        \hline
        \textbf{Notation} & \textbf{Definition} \\
        \hline\hline
        $T$ & Iteration limit. Unknown a priori. \\ 
        \hline
        $t$ & Current iteration count. \\ 
        \hline
        $S_t$ & Set of scores of the running solution at iteration $t$ \textit{or} the solution set itself. \\ 
        \hline
        $k$ & Cardinality LIMIT on the query result. \\ 
        \hline
        $\mathcal{M}(\mathcal{X})$ & The set of all multisets over some domain $\mathcal{X}$. \\ 
        \hline
        $\mathcal{X}$ & The domain of all legal elements. \\ 
        \hline
        $D$ & The target dataset. \\ 
        \hline
        $n$ & Number of elements in $D$. \\ 
        \hline
        $f$ & Scoring function. \\ 
        \hline
        $L$ & Number of clusters in some clustering of $D$. \\ 
        \hline
        $\mathcal{D}_l$ & Distribution of scores of elements in the $l^{\text{th}}$ cluster. \\ 
        \hline
        $p(x | \mathcal{D}_l)$ & Probability of sampling outcome $x$ from $D_l$. \\ 
        \hline
        $S_{(k)}$ & $k^{\text{th}}$ largest element of some (multi)set $S$. \\ 
        \hline
        $\text{STK}(S)$ & Sum of the $k$ largest elements in (multi)set $S$. \\ 
        \hline
        OPT & $\mathbb{E}[STK(S_t)]$ for the optimal adaptive algorithm. \\
        \hline
        $\Delta_t$ & $\text{STK}(S_t) - \text{STK}(S_{t-1})$. \\ 
        \hline
        $\mathbb{E}[\Delta_{t,l}]$ & Expected $\Delta_t$ after sampling from $\mathcal{D}_l$. \\ 
        \hline
        $l_t$ & Cluster chosen by our algorithm at iteration $t$. \\ 
        \hline
        $B$ & Number of buckets in histogram sketches. \\ 
        \hline
        $\alpha$ & Initial maximum range of the histograms. \\
        \hline
        $\beta$ & Re-binning overestimation parameter. \\
        \hline
        $F$ & Frequency to check fallback condition. \\
        \hline
    \end{tabularx}
    \caption{Definition of notations used throughout the paper.}
    \label{tbl:notations}
\end{table}

\subsection{Sum-of-Top-\textit{k} Objective}

We must define an objective for our algorithm. In top-$k$ query problems with stronger assumptions, algorithms aim for high precision or recall under a latency budget. However, precision and recall are extrinsic measures of solution set quality. A solution set does not have precision or recall in a vacuum, and we would also need information about the ground truth solution. Such information is not available in the opaque setting.

In the opaque setting, we need an \textit{intrinsic} measure of solution set quality. We propose Sum-of-Top-$k$ scores (STK) as this measure. Given a finite set of non-negative numbers $S$, STK is the sum of up to $k$ largest elements of $S$. 
If the score of each element represents its utility, then $\text{STK}(S)$ represents the total utility of $S$ for a user that is interested only in the top-$k$ elements.
Formally,
\begin{equation}
    \text{STK}(S) = \sum_{k' = 1}^k (S)_{(k')}
\label{eq:stk}
\end{equation}
where $(S)_{(k)}$ is the $k$th largest element in $S$. 
Note that $S_{(k)}$ stands for the $k$th smallest element in the standard notation for order statistics. We flip the notation for simplicity.

\begin{example}
    Alice queries Bob to retrieve a collection of valuable used car listings from an online automobile trading platform, in descending order of appraised price. Since Alice has limited time and patience, she will read only up to the first 250 rows of Bob's report. Bob's performance will be evaluated based on the total estimated prices of the top 250 cars that he has selected. 
\end{example}

Although STK gives limited guarantees in terms of precision or recall, it is useful for several reasons. Its marginal gain is  easy to compute and to estimate. It has desirable diminishing-returns properties. Maximizing STK exactly gives the ground truth solution, barring ties. Finally, we empirically show that STK is tightly correlated with precision and recall in real-world data.

\subsection{Concrete Problem}

Let $D \in \mathcal{X}^n$ be a dataset over some arbitrary domain $\mathcal{X}$ (e.g., images, documents, rows). Assume an opaque scoring function $f: \mathcal{X} \to [0, \infty)$ with a potentially unknown but fixed maximum range. Let multiset $S_{t}\subset D$ be an agent's running top-$k$ solution at iteration $t$. Note that $S_{t}$ may contain fewer than $k$ elements, at the beginning, when $t<k$. 
We consider the setting where in each iteration $t$, the agent can access an element $x_t \in D$ and evaluate $f(x_t)$. Then, the running solution $S_{t-1}$ is updated based on $f(x_t)$, forming $S_t$. That is, if $f(x_t)>(S_{t-1})_{(k)}$, $x_t$ kicks out  $(S_{t-1})_{(k)}$. 
The running solution may contain ties. 
We evaluate the performance of $S_t$ using $\text{STK}(S_t)$. 
The user monitors the running solution and retrieves the result as soon as satisfied. As such, our goal is to query $D$ in an order that maximizes STK at each iteration.

\subsection{Abstract Problem}

We formalize the problem using a multi-armed bandit (MAB) framework. MAB is a class of problems and algorithms for online decision-making under uncertainty. In the classical MAB setting, an agent interacts with a collection of unknown random distributions, referred to as arms. In each iteration, the agent selects an arm to probe and receives a  random reward from the corresponding distribution. The goal is to maximize the expected sum of rewards over an iteration limit.  
We adapt this framework to our setting. A cluster of data points in a dataset is  an arm with an unknown distribution of scores from the opaque UDF. 
The reward of a sampled point at iteration $t$  is the marginal gain 
$\Delta_t = \text{STK}(S_{t}) - \text{STK}(S_{t-1})$. 
It follows from telescoping sum that $\sum_{t=1}^T \Delta_t = \text{STK}(S_T)$, since $\text{STK}(\emptyset) = 0$. 

\begin{definition}[Top-$k$ Bandit Problem]
    Suppose there are $L$ unknown probability distributions, denoted $\mathcal{D}_1, \dots, \mathcal{D}_L$. Each distribution has a finite domain of non-negative integers $\mathcal{X} \subset \mathbb{Z}^+$. (We will generalize this to real intervals later.) The agent starts with an empty multiset $S_0 = \emptyset$. In each iteration $1 \le t \le T$, the agent samples a value i.i.d. from some distribution $\mathcal{D}_l$ for $l \in [L]$ and adds it to $S_{t-1}$ to obtain $S_t$. The goal is to select the appropriate arms, maximizing 
    $\mathbb{E}[\text{STK}(S_T)]$
    by the iteration limit $T \in \mathbb{Z}^+$. \label{def:abstract} 
\end{definition}

Consider Example~\ref{ex:running}. Alice might index the data using $k$-means clustering and a dendrogram, as in He et al. \cite{he2020method}. This abstraction also applies to existing B-trees or a union of datasets. When working with a hierarchical index, Definition~\ref{def:abstract} applies to clusters in each layer of the tree, as shown in Figure~\ref{fig:overview-visualization}. While the analysis assumes sampling i.i.d. from an abstract distribution $\mathcal{D}_l$, in practice, Alice samples listings from each cluster without replacement and applies $f$ to obtain a score. In this case, the scores are the appraised prices.
Alice monitors the quality of the solution via a live interface and retrieves the final solution once satisfied.

Solving this abstract problem is challenging. A good algorithm must balance exploration and exploitation. Furthermore, even when the distributions are known, the nature of exploitation differs from the standard bandit scenario. As we demonstrate in \S~\ref{sec:analysis}, the bandit must consider the full distribution, in contrast to standard bandit algorithms, which only require estimating the mean reward of each arm. Finally, computational overhead must remain low.

\section{Top-\textit{k} Bandit Algorithm} \label{sec:algorithms}
 
In this section, we design and analyze an algorithm to effectively solve the abstract problem (Definition~\ref{def:abstract}).

\subsection{Theoretical Variant} \label{sec:discrete}

We first describe an algorithm for a simplified setting. Let us ignore the scoring function $f: \mathcal{X} \to [0,\infty)$, and draw a score directly from distribution $\mathcal{D}_l$. Concretely, suppose $\mathcal{X}$ is a finite set of non-negative integers, each $\mathcal{D}_l$ is a probability mass function over $\mathcal{X}$, and $f$ is the identity function. 

\textbf{Greedy Action.} Most bandit algorithms have some notion of a greedy exploitation. What is an analogous objective for the top-$k$ bandit setting? Notice that only the highest $k$ scores are useful. Hence, we cannot adopt the greedy action from the standard bandit setting. Instead, a greedy action in our context is one that maximizes the expected marginal gain in STK. The expected marginal gain from sampling $D_l$ is

\begin{align}
    &\mathbb{E}[\Delta_{t,l}] = \mathbb{E}[\text{STK}(S_t) - \text{STK}(S_{t-1})] \nonumber\\
    &= \sum_{x_t \in \mathcal{X}} p(x_t | D_l) \cdot \mathbb{I}(f(x_t) > (S_{t-1})_{(k)}) \cdot (f(x_t) - (S_{t-1})_{(k)})
    \label{eq:marginal}
\end{align}

where $x_t \sim \mathcal{D}_l$, and $p(x_t | \mathcal{D}_l)$ is the probability of outcome $x_t$ in $\mathcal{D}_l$. 

This equation holds since any new score greater than $(S_{t-1})_{(k)}$ will ``kick out'' $(S_{t-1})_{(k)}$ in the running solution. As a result, we might prioritize an arm with a fat tail and a lower mean over another arm with a thinner tail and a higher mean.

\textbf{Exploration-Exploitation Tradeoff.} The term $(S)_{(k)}$ can exhibit arbitrarily large local curvature. To illustrate, let $k = 2$ and $\mathcal{X} = {0, 1, \dots, 100}$. Further, let $S_2 = {0, 0}$, $S_3 = {0, 0, 100}$, and $S_4 = S_5 = {0, 0, 100, 100}$. Then, the marginal increase in $S_{(k)}$ is first 0, then 100, then 0. Consequently, there is no guarantee that an arm profitable in one iteration will remain profitable in the next, and vice versa. Therefore, an Upper Confidence Bound (UCB) style algorithm, which prioritizes exploring promising arms, is unsuitable. Instead, we explore each arm uniformly. 

\textbf{Algorithm description.} With these considerations, our proposed algorithm for the discrete domain is an $\varepsilon$-greedy algorithm. In each iteration, with some probability $\varepsilon_t$, the agent explores a uniformly random arm. Otherwise, it exploits by choosing the arm that maximizes $\mathbb{E}[\Delta_{t,l}]$, as formulated in Equation~\ref{eq:marginal}. $\mathbb{E}[\Delta_{t,l}]$ is estimated by maintaining certain statistics. Let $N_l$ be the number of times arm $l$ has been visited, for all $l \in [L]$. Let $N_{l, x}$ be the number of times $\mathcal{D}_l$ was probed and returned a score of $x$, for all $l \in [L]$ and $x \in \mathcal{X}$. The greedy arm $l^G_t$ is given by
\begin{equation}
l^G_t = \underset{l \in [L]}{\mathrm{argmax}} \sum_{x \in \mathcal{X}} \frac{N_{l,x}}{N_l} \max \left(x - (S_{t-1})_{(k)}, 0 \right).
\end{equation}

We set the exploration chance to decrease at a rate of $\varepsilon_t = t^{-1/3}$, since exploration to refine empirical probability estimates is most valuable early on.

\begin{algorithm}[htb]
    \caption{Practical $\varepsilon$-Greedy for Top-$k$ MAB}\label{alg:cont-epsgreedy}
    \KwData{Bucket count $B$, cardinality constraint $k$, distributions $D_1, \cdots, D_L$, initial maximum $\alpha$, scaling parameter $\beta \in [1, \infty)$, fallback frequency $F \in [0, 1]$}.
    \KwResult{$k$ highest scoring elements at iteration limit $T > 0$.}
    $H_l \gets $ empty uniform histogram with range $[0, \alpha]$ and $B$ buckets for all $l \in [L]$\;
    $PQ \gets $ priority queue with cardinality constraint $k$\; 
    \For{$1 \le t \le T$}{
        \If{after $0.3\times n$ iterations and every $F \times n$ iterations}{
            Check and apply fallback condition, if applicable\;
        }\ElseIf{$\text{with probability } t^{-1/3}$}{
            $l_t \gets \text{uniformly random from } [L]$
        }\Else{
            $l_t \gets $ arm that maximizes $\mathbb{E}[\Delta_{t,l}]$\;
        }
        $x_t \gets f(\mathrm{sample}(D_{l_t}))$\;
        \If{$x_t > \text{maximum range in } H_{l_t}$}{
            Extend $H_{l_t}$'s range to $[0, \beta \cdot x_t]$\;
        }
        Insert $(x_t, s_t)$ to $PQ$.\;
        \If{$PQ_{(k)} > \text{upper border of second lowest bin}$}{
            \text{Extend lowest bucket for} $H_{l_t}$
        }
        Update $H_{l_t}$ with $x_t$\;
        \If{$D_{l_t}$ is empty}{
            Drop $D_{l_t}$\;
        }
    }
    \Return $PQ.\text{Values}()$
\end{algorithm}

\subsection{Practical Variant}
\label{sec:practicalvar}

Although the algorithm described in \S~\ref{sec:discrete} is most amenable for analysis, we do not use it in practice. Instead, we implement Algorithm~\ref{alg:cont-epsgreedy}, which generalizes to arbitrary domain $\mathcal{X}$ and continuous scores while being efficient in practice. 

\subsubsection{Algorithm Description} In this version, we maintain a priority queue of the $k$ highest scores seen so far. 
The queue is implemented using a cardinality-constrained min-max heap~\cite{atkinson1986min}. 
The algorithm uses histograms to model the continuous distribution $D_l$. A histogram stores three pieces of information: the number of visits to each arm, the bin borders, and the number of times each bin has been sampled. It initializes the histograms as empty equi-width histograms with range $[0, \alpha]$ for some small constant $\alpha$ and $B$ bins. During each exploitation step, the algorithm computes $\mathbb{E}[\Delta_{t,l}]$ under the uniform value assumption, interpreting each histogram as a piecewise-continuous function that is constant within each bin.
Ties are broken randomly.

\subsubsection{Indexing scheme} \label{sec:index}

So far, we assumed  the existence of a flat index that clusters the elements. In practice, we adopt the \text{VOODOO} index from He et al.~\cite{he2020method}, which creates a hierarchical clustering of the search domain based on a task-independent and inexpensive vectorization scheme. 
The intuition is that elements with similar vector representations will likely have similar scores. Hence, grouping similar elements into a cluster allows us to exploit this statistical information~\cite{anderson2016input}. The hierarchy groups similar clusters into the same subtree to further benefit from the heuristic~\cite{he2020method}. Our goal is to exploit any useful statistical information captured by the index and to incur low overhead in the worst case.

At a high level, index construction has three phases: vectorization, clustering, and tree building.
We vectorize images using pixel values. For tabular data, we impute and normalize numeric and boolean columns. We then apply $k$-means clustering over the elements' vector representations. 

We take a subsample for clustering if the dataset is large. Finally, we build a dendrogram of the cluster centroids using hierarchical agglomerative clustering (HAC) with average linkage. 
The index construction process is fast, with a nearly linear $O(n L^3)$ time complexity, since $L \ll n$, and a small constant factor thanks to cheap vectorization methods.

Each cluster in the index is an arm. 
Similar to He et al., we run our bandit algorithm over clusters in each layer of the index. The histogram of each cluster approximates the scores of the UDF for all points in its descendant clusters. Upon selecting a cluster, its children constitute the collection of arms that the agent can pull in the next bandit loop.  

\subsubsection{Fallback strategy} \label{sec:fallback}

The intuition behind the index may not always hold. There are two main failure modes: the tree might be ineffective, or the clustering itself may not contain useful information. To cope with the failure modes, He et al. proposed dynamic index reconstruction and scan fallback mode. We design an analogous fallback  strategy.

In our scenario, as our objective function is not monotone, an index that was useful in earlier iterations may no longer be useful in later iterations. As such, we check the failure conditions periodically as opposed to only once. We first process 30\% of the dataset without checking the failure condition to ensure that histogram sketches are reasonably accurate. Then, we re-test the fallback condition after every 1\% of the dataset has been processed.

The tree index is not effective if the greedy arm is not the arm chosen by the hierarchical bandit during a greedy iteration. 
This may happen if the greedy arm is in the same subtree as some other ``bad'' arms~\cite{he2020method} 
To test this condition, we first compute estimated marginal gains of each leaf cluster using the histograms. 
The cluster with the highest marginal gain is the greedy arm. 
Then, we simulate the hiararchical bandit navigating down the tree index, choosing the greedy child in each layer. 
If the leaf cluster chosen by this procedure is not the greedy arm, then the fallback condition holds. 
If so, we turn the index into a flat partition, removing the tree while preserving the clustering.

\textbf{Clustering fallback.}

The clustering is not effective if greedy exploitation would result in a lower slope in the STK versus time curve than a simple uniform sample. 
We fall back to a uniform sample rather than a linear scan, as the former is better suited for the any-time query model. 
We estimate the slope of the two competing algorithms at iteration $t$ as follows. 

\begin{equation*}
    \text{slope}_{\text{bandit}} \approx \frac{\mathrm{argmax}_{l \in L}\; \mathbb{E}[\Delta_{t,l}]}{\text{scoring function latency} + \text{bandit latency}},
\end{equation*}
\begin{equation*}
    \displaystyle
    \text{slope}_{\text{sample}} \approx \frac{\sum_{l \in [L]} |D_l| \times \mathbb{E}[\Delta_{t, l}]}{(\sum_{l \in L} |D_l|) \times \text{scoring function latency}}.
\end{equation*}
Here $|D_l|$ is the remaining size of the $l^{\text{th}}$ cluster and $\mathbb{E}[\Delta_{t, l}]$ is the expected marginal gain in STK after sampling from the $l^{\text{th}}$ cluster. Scoring function latency and bandit latency are measured dynamically. 
If $\text{slope}_{\text{sample}} > \text{slope}_{\text{bandit}}$, then we shuffle all remaining elements, then scan from the shuffled list.

\subsubsection{Advanced histogram maintenance} \label{sec:histmaintenance}

Our algorithm makes the uniform value assumption, which does not always hold. To remedy discrepancies between the histogram sketch and the true distribution, we implement several histogram maintenance strategies.

\textbf{Re-binning.} There are two issues with the histogram approach. First, $S_{(k)}$ increases over iterations. Consequently, information in the lower bins of the histogram becomes irrelevant, whereas the higher ranges of the domain require greater precision. Second, if the maximum value in the domain exceeds the initial estimate, the histogram may fail to capture important distributional information. While increasing bin count could alleviate both problems, we aim to ensure bounded overhead. Instead, we implement methods to extend the lowest bin of the histogram (Figure~\ref{fig:histogram-2}) and to extend the total range of the histogram (Figure~\ref{fig:histogram-3}). Both methods use the uniform value assumption. We re-bin the lowest bin when $S_{(k)}$ is greater than the upper limit of the second lowest bin. This means that the two lowest bins can be combined and still convey the same amount of useful information. We use the parameter $\beta$ to slightly overestimate the true maximum range when a large score is sampled, with a default value of $\beta = 1.1$. Any $\beta \in [1, 2]$ should be reasonable.

\textbf{Empty child handling.} In the analysis, we assume that the size of each arm is infinite. This assumption does not hold if $T$ is large or if some arms have few elements. Thus, a leaf cluster in the index could run out of fresh elements. In such cases, we drop the leaf node and recursively drop its parent if it no longer has any children. This raises a potential issue: the parent cluster's histogram may now be outdated, potentially causing performance degradation if one of the parent's children was ``good'' and the other was ``bad.'' To remedy this issue, we implement a method to subtract one histogram from another using the uniform value assumption. As an edge case, the child may have different bin borders than the parent, resulting in bins with negative counts. We always round up the histogram's bin counts to zero if they become negative, to avoid numerical issues.

\subsubsection{Batched processing} Algorithm~\ref{alg:cont-epsgreedy} assumes that the histograms are updated after each iteration. In practice, batching model inference yields significant speedups by amortizing GPU latency and parallelizing computation.
Batching complicates the exploration rate guarantees of our algorithm, but we find that dividing $t$ by the batch size suffices.

In extreme cases, the optimal batch size could be large compared to the number of elements in an arm. This scenario occurs if many elements fit into the GPU memory.
In such cases, our bandit algorithm may not benefit from much statistical optimization, and the \scan baseline could be much faster. In our experiments, our batch sizes are small compared to the size of arms, so this extreme case does not arise. 

\subsubsection{Implementation details} Our proposed algorithm is implemented as a standalone library, written in Python. We assume that the index fits into memory, which holds for all datasets used in our experiments. Thus, a simple JSON file suffices as the index. We further assume that each element in the search domain has a unique string ID. These assumptions can be relaxed with future engineering efforts. In our library, a user-defined ``sampler'' function takes an ID and additional parameters as input, and returns an object---the element itself---of arbitrary type. Another user-defined scoring function takes the object as input and returns a non-negative integer. These functions also accept an array of inputs for batching. We write our scoring functions using standard Python data science libraries, such as PyTorch. 

\begin{figure}
    \centering
    \subfloat[]{
        \includegraphics[width=0.33\columnwidth]{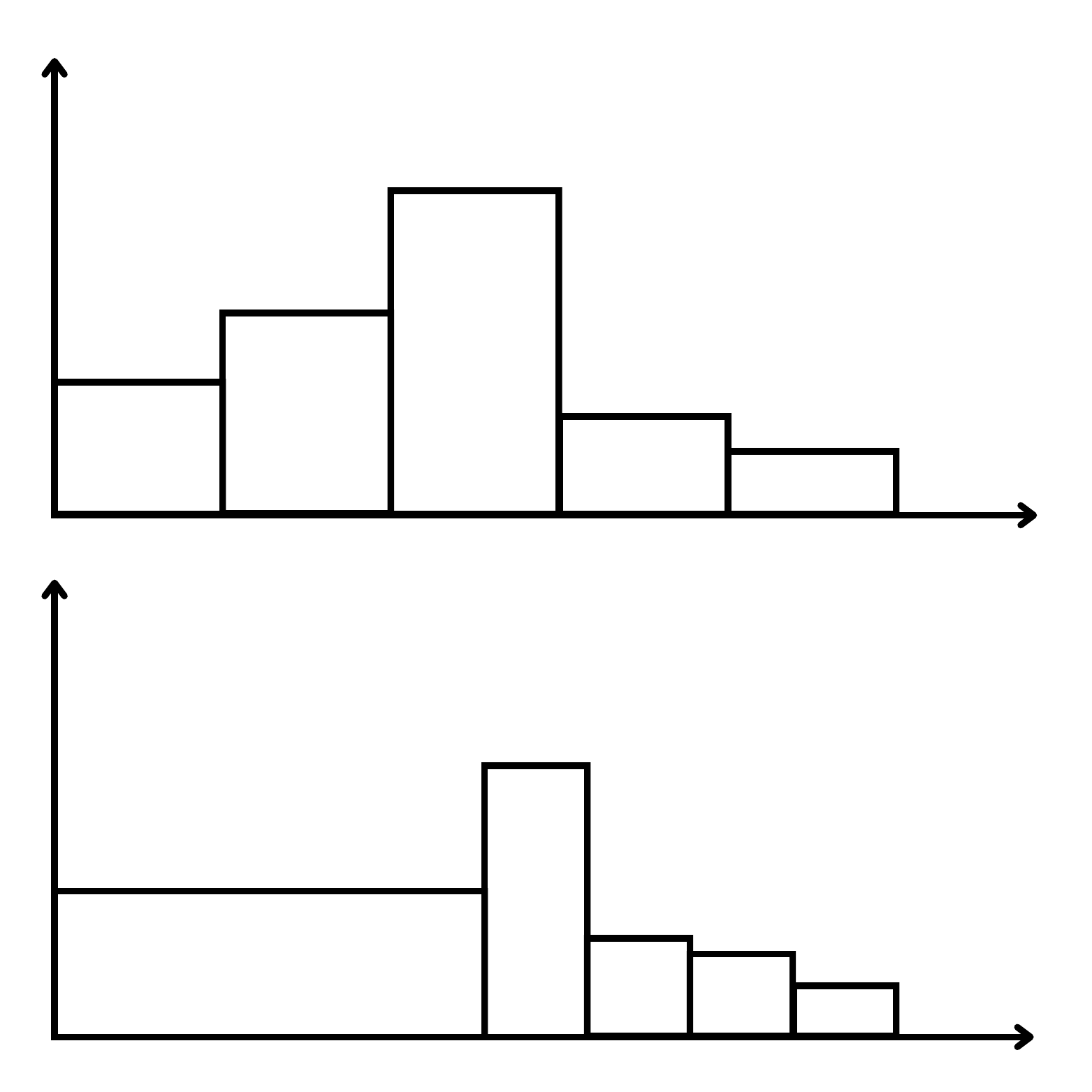}
        \label{fig:histogram-2}
    }
    \subfloat[]{
        \includegraphics[width=0.33\columnwidth]{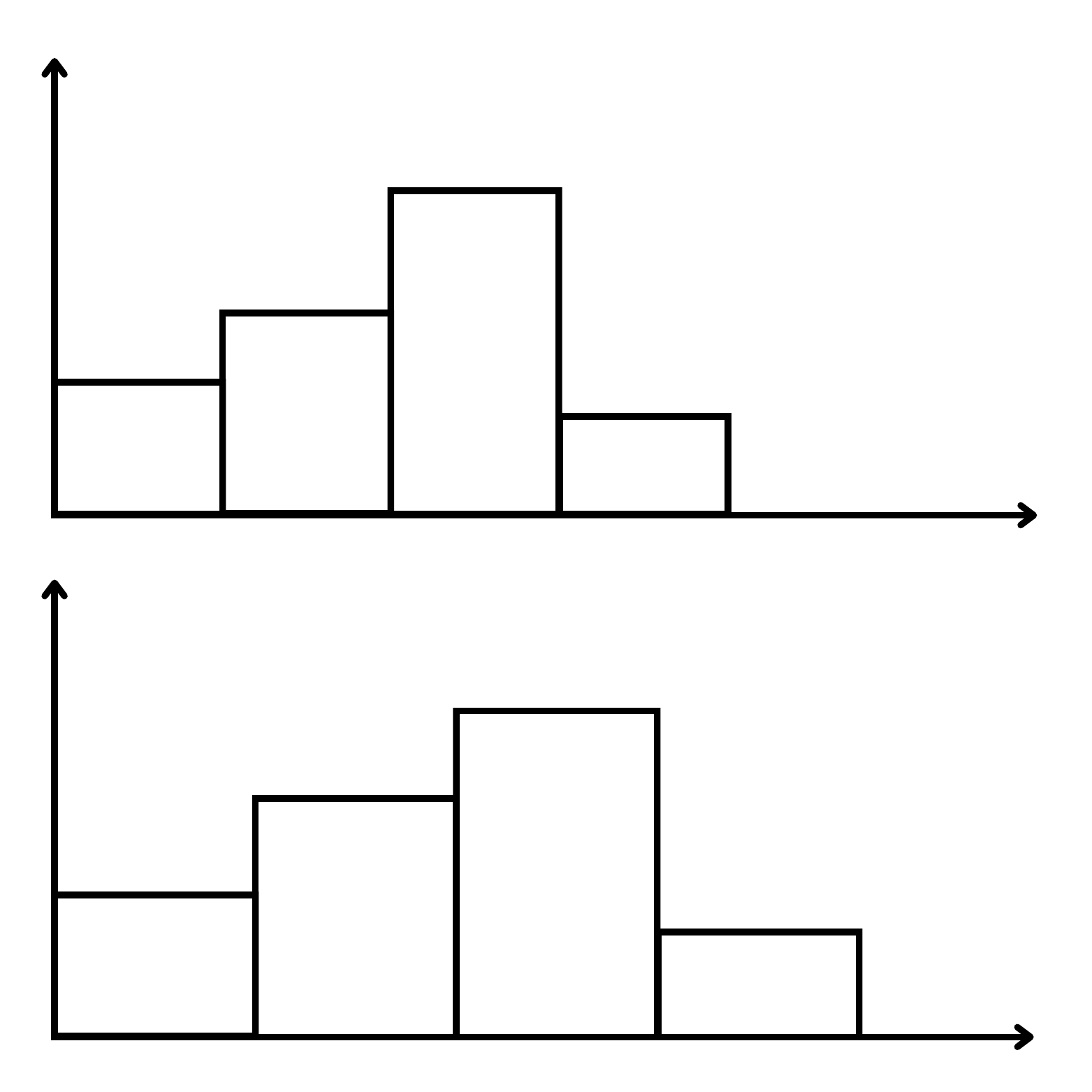}
        \label{fig:histogram-3}
    }
    \subfloat[]{
        \includegraphics[width=0.33\columnwidth]{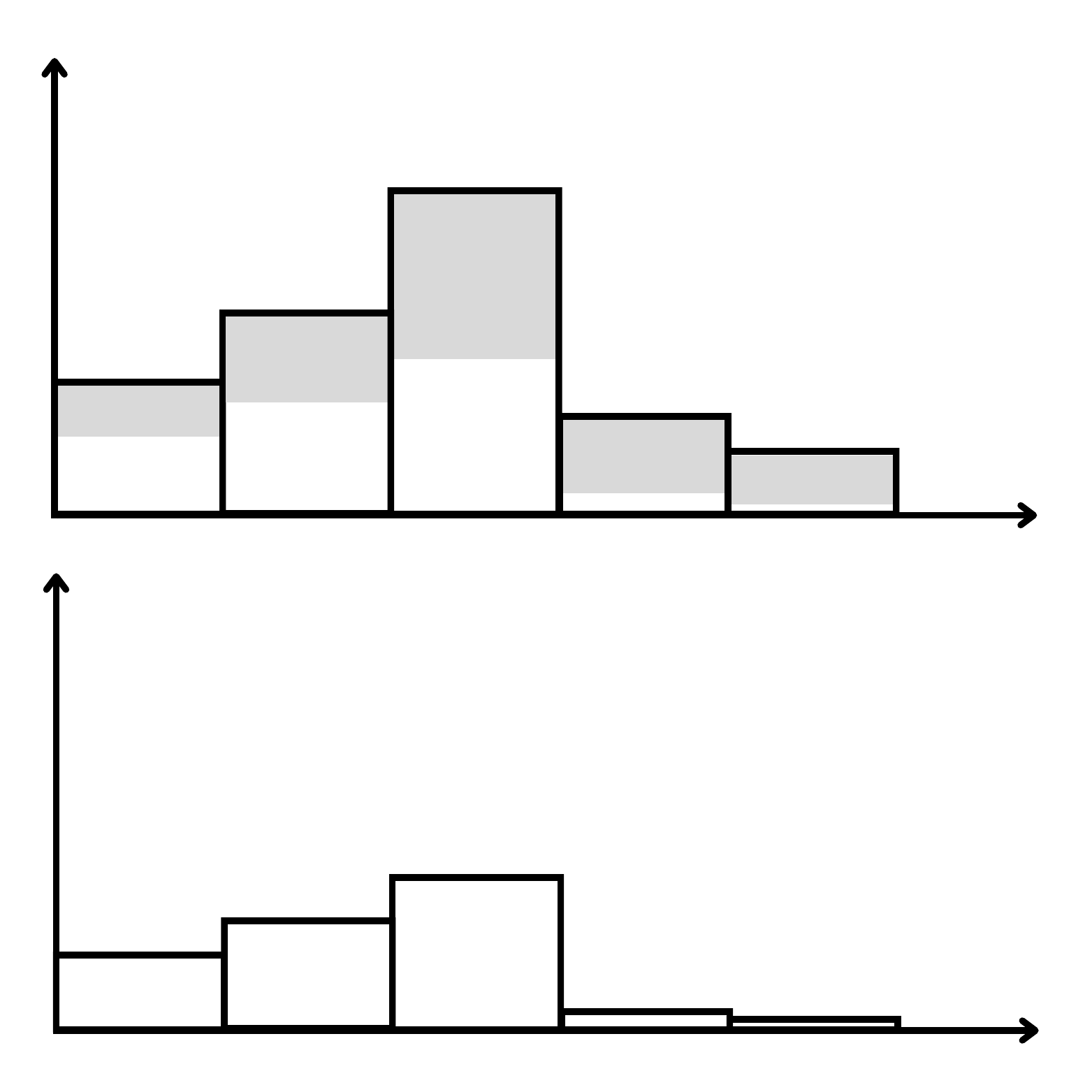}
        \label{fig:histogram-4}
    }
    \caption{A visualization of our histogram maintenance strategy. (a) Extending the lowest bin as $R_{(k)}$ value grows. (b) Extending the range to contain unexpectedly large scores, and (c) subtracting a histogram (shaded) from another histogram.} 
    \Description{(a) The lowest several bins of the histogram is consolidated into one, and the remaining bins are split into smaller bins. (b) The bins of a histogram are lengthened out to increase the maximum range of legal values. (c) The height of a histogram is decreased equal to the amount subtracted.}
    \label{fig:histogram}
\end{figure}

\subsubsection{Summary and End-to-End Example}

We now return to Example~\ref{ex:running}, and review the full practical workflow.

\begin{example}
    Alice hypothesizes that used cars that have similar feature values will have similar valuations. As such, she clusters listings based on their feature values using $k$-means clustering. She then builds a tree that groups close clusters into the same subtree. Then, she deploys the following algorithm. 
    \begin{enumerate}
        \item It initializes an empty histogram for all nodes in the tree, including leaf clusters. It also initializes a priority queue that has a maximum capacity of at most $k$ elements. The queue will contain (score, element) tuples. 
        \item In each iteration, the algorithm chooses a leaf cluster by invoking the $\varepsilon$-greedy bandit algorithm on each layer of the tree index, which:
        \begin{enumerate}
            \item With a decreasing probability, chooses a random child. 
            \item Otherwise, chooses the child that greedily maximizes marginal gain in STK. 
            \item This process repeats until it reaches a leaf cluster. 
        \end{enumerate}
        \item It takes a sample from the chosen leaf and computing the sample's score. 
        \item It updates the histograms of the leaf node and all of its parent nodes in the tree. The histogram class adjusts its bin borders when $S_{(k)}$ becomes too large or if a score that exceeds its limit is sampled. 
        \item After an initial threshold, the algorithm periodically checks the failure conditions. It may fall back to a flat index or a uniform sample over the remaining elements. 
        \item When Alice is satisfied with the solution quality, she terminates and retrieves the elements in the priority queue. 
    \end{enumerate}
\end{example}

\section{Analysis} \label{sec:analysis}

In this section, we establish a theoretical regret bound. 
Our main result (Theorem~\ref{thm:regret}) states that Algorithm~\ref{alg:cont-epsgreedy} finds a solution whose STK is lower bounded by $(1-\exp(-1-1/2T)) \text{OPT} - O(T^{2/3})$ in expectation.
That is, the quality of running solution approaches a constant factor of $1-1/e \approx 0.63$ with a multiplicative and an additive regret.

To show this, we will first assume that all $L$ distributions of arms are known. 
Under this assumption, we will reduce the problem of assigning budgets to each arm as a stochastic diminishing returns (DR) submodular maximization problem (\S~\ref{sec:algorithm-monotone-dr}). 
Intuitively, assigning a larger budget increases the expected quality of the solution, but larger budgets have diminishing returns.
A simple greedy algorithm achieves $(1-1/e)$-approximation to the optimal in expectation~\cite{asadpour2016maximizing}. 
Finally, we generalize the result to unknown distributions (\S~\ref{sec:algorithm-regret}). 
We show that our algorithm, with high probability, chooses a nearly optimal arm. 
We then apply this to a standard formula for analyzing submodular maximization algorithms. 

\subsection{Objective Function Analysis} \label{sec:algorithm-monotone-dr}

In the classical MAB setting, the reward obtained from a distribution contributes directly to the overall objective. 
Thus, the greedy action is optimal. 
In Top-$K$ MAB, however, the greedy action may be suboptimal. 
Let there be distributions $D_1$ and $D_2$, where $D_1$ has a higher mean than $D_2$, but $D_2$ has a higher variance. 
Suppose we are about to make the first sample ($t = 1$). 
If $T \gg k$, it may be optimal to sample from $D_2$, even though $D_1$ has a higher mean. 

The main objective of this section is to establish that the value of a greedy action is within a constant factor to optimal. 
The key idea is that the STK objective exhibits a diminishing returns property. Specifically, it is monotone and DR-submodular~\cite{soma2018maximizing} w.r.t. the sampling budget allocated to each cluster. Then, we reduce the top-$k$ sampling problem, under known distributions, as the \emph{stochastic monotone submodular maximization problem under matroid constraint}, which is NP-hard but admits a $(1-1/e)$-approximate adaptive greedy algorithm~\cite{asadpour2016maximizing}. 

\noindent\textbf{Preliminaries.} Let us define the desirable properties. 
Let $S \in \mathcal{M}(\mathbb{Z})$ be a multiset of integers, and let $f: \mathcal{M}(\mathbb{Z}) \to \mathbb{Z}^+$ be a function that returns a number given a multiset. 
An example of such function $f$ is the STK function (Equation~\ref{eq:stk}).
Let us denote the multiplicity of element $x$ in multiset $S$ as $m_S(x)$. Given two multisets $X, Y \in \mathcal{M}(\mathcal{X})$, we say that $X \le Y$ when $m_X(x_i) \le m_Y(x_i)$ for all $x_i \in \mathbb{Z}$. 
For example, $\{0, 1\} \le \{0, 0, 1, 1, 1\}$, but $\{0, 0, 1\}$ and $\{0, 1, 1\}$ are incomparable.
The $\le$ operator defines a lattice over multisets. 

\noindent Then, $f$ is \emph{monotone} iff
\begin{equation}
    f(X) \le f(Y) \quad
\end{equation}
for all $X, Y \in \mathcal{M}(\mathbb{Z})$ where $X \le Y$. 

\noindent Next, a function $f$ is \emph{DR-submodular over integer lattice} iff
\begin{equation}
    f(X + x_i) - f(X) \ge f(Y + x_i) - f(Y)
\end{equation}
for all $X, Y \in \mathcal{M}(\mathbb{Z})$ s.t. $X \le Y$, and $x_i \in \mathbb{Z}$. 

\noindent\textbf{Analysis of STK function.} With the definitions established, we now show that the STK objective function is monotone and DR-submodular. 
Intuitively, adding new elements to some set never decreases its STK. Furthermore, adding an element to a small set has potential to significantly improve its STK, whereas adding the same element to a larger set improves STK by a smaller amount.

\begin{theorem}\label{thm:g-monotone-dr}
    $\text{STK}(S)$ is monotone and DR-submodular.
\end{theorem}

\begin{proof}
    Let $x \in \mathcal{X}$. Then, 
    \begin{equation}\text{STK}(S \cup \{x\}) = \begin{cases}
        \text{STK}(S) + x - S_{(k)} &\text{if } x \ge S_{(k)}\\
        \text{STK}(S) &\text{otherwise.}
    \end{cases}\end{equation}
    Since adding an element to $S$ does not decrease \text{STK}, that is $\text{STK}(S \cup \{x\}) \ge \text{STK}(S)$,  the function $\text{STK}(S)$ is monotonic. Next, let $S_1, S_2$ be multisets such that $S_1 \le S_2$. Then, we  prove that 
    \begin{equation}
        \text{STK}(S_1 \cup \{x\}) - \text{STK}(S_1) \ge \text{STK}(S_2 \cup \{x\}) - \text{STK}(S_2).
        \label{eq:det-dr}
    \end{equation}
    This is equivalent to the following. 
    \begin{equation}
        \mathbb{I}(x \ge {(S_1)}_{(k)}) ( x - {(S_1)}_{(k)} ) \ge \mathbb{I}(x \ge {(S_2)}_{(k)}) ( x - {(S_2)}_{(k)} )
        \label{eq:det-dr-2}
    \end{equation}
    
\noindent Clearly ${(S_1)}_{(k)} \le {(S_2)}_{(k)}$ since $S_1 \le S_2$. Hence, 
    \begin{equation}
        x - {(S_1)}_{(k)} \ge x - {(S_2)}_{(k)} \ge 0. \label{eq:g-dr-1}
    \end{equation}
    By total ordering of $\mathbb{Z}$, either $x < {(S_1)}_{(k)}$, ${(S_1)}_{(k)} \le x < {(S_2)}_{(k)}$, or ${(S_2)}_{(k)} \le x$. Hence, 
    \begin{equation}
        \mathbb{I}(x \ge {(S_1)}_{(k)}) \ge \mathbb{I}(x \ge {(S_2)}_{(k)}) \ge 0. \label{eq:g-dr-2}
    \end{equation}
    By Equation~\ref{eq:g-dr-1} and Equation~\ref{eq:g-dr-2}, Equation~\ref{eq:det-dr-2} always holds, proving DR-submodularity. 
\end{proof}

\noindent\textbf{Analysis of budget oracle.} 
A \textit{budget} restricts how many times we can sample from each arm or cluster. 
For example, suppose are just two clusters A and B that divide a dataset in half. A valid budget might be $\{(A, 2), (B, 10)\}$ which means to sample from A twice and sample from B 10 times.
Assume   $\text{BS}: (\mathbb{Z}^+_0)^L \to \mathbb{R}$ is a function that evaluaes a budget allocation. 
It takes as input a mapping from each of the $L$ arms/clusters to a non-negative integer. 
Then, it returns the value of that budget. 
Let $X$ be a budget, and let $X_l$ be the budget allocated to $\mathcal{D}_l$.
We define $S_r$ to be a random variable representing the outcome of the following procedure. 

\begin{proc}
    Start with multiset $S = \emptyset$. Then, sample the $l$\tth{} cluster $x_l$ times independently, for each $l \in [L]$. Add all scores to the running set of scores $S$ and return $S$. We define $\text{BS}(X)$ to be 
    \begin{equation}
        \text{BS}(X) = \mathbb{E}[\text{STK}(S_r)].
    \label{eq:bs}
    \end{equation}  
    \label{def:greedy-budgeted-sample}
\end{proc}

Then, the following theorem holds. 
\begin{theorem}\label{thm:h-monotone-dr}
    $\text{BS}(X)$ is monotone and DR-submodular. 
\end{theorem}
Intuitively, sampling more never hurts. Sampling one more time is highly beneficial if the existing budget is tight, and less beneficial if the existing budget is plentiful.

\begin{proof}
    Analogous to multisets, two integer vectors $X, Y$ satisfy $X \le Y$ iff $x_i \le y_i$ element-wise. 
    Let $X, Y \in (\mathbb{Z}^+)^L$ be budget allocations s.t. $X \le Y$.
    For each $l$, suppose there is a pre-generated infinite tape of i.i.d. samples from $D_l$. 
    Let there be two agents: The first reads the first $x_l$ elements from the $l$th tape for each $l \in [L]$, and the second reads the first $y_l$ instead. 
    Then, regardless of the values realized on the tapes, we have $S_X \subseteq S_Y$. By monotonicity of $\text{STK}$ (Theorem~\ref{thm:g-monotone-dr}), $\text{STK}(S_X) \le \text{STK}(S_Y)$. 
    Since the inequality holds in all outcomes, regardless of the realizations on the tapes, we have 
    \begin{equation}
        \mathbb{E}[\text{STK}(S_X)] \le \mathbb{E}[\text{STK}(S_Y)]
    \end{equation}
    which proves the monotonicity of $\text{BS}$. 

    Next, let $X, Y \in (\mathbb{Z}^+)^L$ such that $X \le Y$, and let some $i \in [L]$. Then, we devise a similar experiment. 

    \begin{proc}\label{def:dr-experiment}
        For each $l \in [L]$, initialize an infinite tape of i.i.d. samples from the distribution $D_l$. 
        Initialize multisets $S_{X'}, S_{X}, S_{Y'}, S_Y$ as $\emptyset$. 
        These multisets have budgets $X + \chi_i, X, Y + \chi_i$, and $Y$, respectively. 
        Consider one multiset and its corresponding budget vector $Z$. 
        Then, add the first $z_l$ elements of the $l$th infinite tape for all $l \in [L]$. 
        Repeat for each of the four multisets. 
    \end{proc}
    By construction of Procedure~\ref{def:dr-experiment}, we have 
    $S_{X'} \ge S_X$, $X_{Y'} \ge S_Y$, and $S_Y \ge S_X$. 
    By DR-submodularity of $\text{STK}$ (Theorem~\ref{thm:g-monotone-dr}), 
    \begin{equation}
        \text{STK}(S_{X'}) - \text{STK}(S_X) \ge \text{STK}(S_{Y'}) - \text{STK}(S_Y).
    \end{equation}
    Since this holds in every possible outcome, 
    \begin{equation}
        \mathbb{E}[\text{STK}(S_{X'}) - \text{STK}(S_X)] \ge \mathbb{E}[\text{STK}(S_{Y'}) - \text{STK}(S_Y)].
    \end{equation}
    By linearity of expectation, 
    \begin{equation}
        \mathbb{E}[\text{STK}(S_{X'})] - \mathbb{E}[\text{STK}(S_X)] \ge \mathbb{E}[\text{STK}(S_{Y'})] - \mathbb{E}[\text{STK}(S_Y)].
    \end{equation}
    By construction of Procedure~\ref{def:dr-experiment}, 
    \begin{equation}
        \mathbb{E}[\text{STK}(S_{X'})] = \text{BS}(X + \chi_i), \quad \mathbb{E}[\text{STK}(S_X)] = \text{BS}(X),
    \end{equation}
    \begin{equation}
        \mathbb{E}[\text{STK}(S_{Y'})] = \text{BS}(Y + \chi_i), \quad \mathbb{E}[\text{STK}(S_Y)] = \text{BS}(Y).
    \end{equation}
    By substitution, 
    \begin{equation}
        \text{BS}(X + \chi_i) - \text{BS}(X) \ge \text{BS}(Y + \chi_i) - \text{BS}(Y)
    \end{equation}
    which proves DR-submodularity of $\text{BS}$. 
\end{proof}

Our final task is establishing $(1-1/e)$-approximation of the adaptive greedy procedure. 

\begin{remark}\label{thm:reduction}
    The top-$k$ bandit problem with known distributions reduces to a monotone submodular maximization problem under uniform matroid constraint.
\end{remark}

We let the objective function be $\text{BS}(X)$ defined in Equation~\ref{eq:bs}, which we proved to be monotone and submodular (Theorem~\ref{thm:g-monotone-dr}). 
The cardinality constraint $T$ is a uniform matroid. 
Hence, maximizing $\text{BS}(S)$ under uniform matroid constraint also maximizes the top-$K$ bandit with known distributions. 

\begin{corollary}\label{thm:adaptive-greedy}
    Adaptive greedy, which chooses distribution that greedily maximizes marginal gain of $\text{STK}(S)$ in each iteration, achieves $(1-1/e)$-approximation to the optimal adaptive algorithm.
\end{corollary}

Corollary~\ref{thm:adaptive-greedy} holds as a direct application of Theorem~4 in~\cite{asadpour2016maximizing}. 

\subsection{Regret Analysis}\label{sec:algorithm-regret}

Corollary~\ref{thm:adaptive-greedy} in \S~\ref{sec:algorithm-monotone-dr} establishes that, if distribution of all arms 
are known, then the adaptive greedy procedure achieves constant approximation. Of course, we cannot assume that the distributions are known. In this section, we show the following theorem. 

\begin{theorem}
    The expectation of STK by Algorithm~\ref{alg:cont-epsgreedy} at iteration limit $T$ is lower bounded by $\left(1 - e^{-1-\frac{1}{2T}}\right)\text{OPT} - \tilde{O}\left(T^{2/3}\right)$ if there is a finite collection of arms and the scoring domain is a finite set of non-negative integers. 
    \label{thm:regret}
\end{theorem}

Our proof strategy combines standard $\varepsilon$-greedy bandit analysis technique with the proof of Theorem~4 in~\cite{asadpour2016maximizing}. 
We divide all outcomes into the clean event and the dirty event. 
In the clean event, the algorithm estimates all bin values to high precision. 
The dirty event is the complement of the clean event. 
Since the dirty event occurs with low probability, the algorithm closes the gap between $\mathrm{OPT}$ and $\text{STK}(S_t)$ by a factor of $\left(\frac{1}{T} - \tilde{O}\left(T^{-1/3}\right)\right)$ due to monotone submodularity. 

\begin{proof}
    The exploration rate for Algorithm~\ref{alg:cont-epsgreedy} at iteration $t$ is given by $\varepsilon_t = t^{-1/3}$. Then by integration, the total expected number of exploration rounds by iteration $T$ is $t^{2/3}$. Let $X_t$ be the number of exploration rounds by iteration $t$, and let $X_t^l$ be the number of explorations on arm $D_l$ by iteration $t$. By the repeated applications of the Hoeffding bound, the following statements hold. 

    The total number of exploration rounds is close to expected with high probability. 
    \begin{equation}
        P\left[ \left| X_t - \mathbb{E}[X_t]\right| \le \tilde{O}\left(t^{-1/3}\right) \right] \le 1-t^{-2}
    \end{equation} 
    A similar bound holds for the number of exploration rounds per arm. We are multiplying the probability $1+L$ times to ensure that the total exploration rounds are close to expectation and that all $L$ arms' exploration rounds are close to expectation. 
    \begin{equation}
        P\left[ \left|X_t^l - \mathbb{E}[X_t^l]\right| \le \tilde{O}\left(t^{-1/3}\right)\; \forall l \right] \le {\left( 1-t^{-2}\right)}^{1+L}
    \end{equation} 
    
    Similarly, we upper bound the probability that the empirical probability estimate for each outcome $x_i$ for each arm is inaccurate compared to the ground truth. We denote all empirical estimates with an overline. This time, multiply the RHS an additional $LB$ times to ensure the bound holds for each arm-outcome pair. 
    \begin{equation}
        P\left[ \left| \overline{P}[X_l = x_i] - P[X_l = x_i] \right| \le \tilde{O}\left(t^{-1/3}\right) \; \forall l \; \forall b \right] \le {\left(1-t^{-2}\right)}^{1+L+LB} \label{eq:event-3}
    \end{equation}
    In Equation~\ref{eq:event-3}, $X_l$ is the random variable representing a sample from $D_l$, $P[X_l = x_i]$ is the probability that $X_l$ is a particular outcome $x_i$, and $\overline{P}[X_l = x_i]$ is the estimate of the probability using samples. Note that RHS of Equation~\ref{eq:event-3} is equivalent to $1-\tilde{O}\left(t^{-2}\right)$. 

    If Equation~\ref{eq:event-3} holds, then clearly
    \begin{equation}
        \left| \overline{\mathbb{E}[\Delta_{t,l}]}(S_{t-1}) - \mathbb{E}[\Delta_{t,l}](S_{t-1}) \right| \le \tilde{O}\left(t^{-1/3}\right).
        \label{eq:gain-approx}
    \end{equation}
    where $\mathbb{E}[\Delta_{t,l}]$ is the expected marginal gain in STK for sampling from $D_l$ and $S_{t-1}$ is the set obtained by the greedy procedure up to the previous iteration. 

A standard result in the known distribution, adaptive setting is
    \begin{equation}
        \mathbb{E}[\Delta_{t}] \ge \frac{1}{T}(\mathrm{OPT} - \text{STK}(S_t))
    \end{equation}
    where $\mathrm{OPT}$ is the expected score of the optimal algorithm~\cite{asadpour2016maximizing}. In the top-$k$ bandit setting, by Equation~\ref{eq:event-3} and~\ref{eq:gain-approx},
    \begin{equation}
        \mathbb{E}[\Delta_{t+1}] \ge \left(1-O\left(t^{-2}\right)\right) \cdot  \left( \frac{1}{T} (\mathrm{OPT} - \text{STK}(S_t)) - \tilde{O}\left(t^{-1/3}\right) \right)
    \end{equation}
    since with probability at least $1-O\left(t^{-2}\right)$, Algorithm~\ref{alg:cont-epsgreedy} chooses an arm with marginal gain at most $\tilde{O}\left(t^{-1/3}\right)$ smaller than the arm chosen by adaptive greedy with known distributions. 

    Then, by algebraic simplification,
    \begin{align}
        \mathbb{E}[\Delta_{t+1}] &\ge \frac{1}{T}(\text{OPT} - \text{STK}(S_t))\left(1 - O\left(t^{-2}\right)\right) - \tilde{O}\left(t^{-1/3}\right).
    \end{align}

    Next, let $\delta_t = \text{OPT} - \text{STK}(S_t)$. Then since 
    \begin{equation}
        \delta_{t+1} = \delta_t - \mathbb{E}[\Delta_{t+1}],
    \end{equation}
    by substitution,
    \begin{equation}
        \delta_t - \delta_{t+1} \ge \frac{1}{T} \delta_t \left(1 - O\left(t^{-2}\right)\right) - O\left(t^{-1/3}\right)
    \end{equation}
    which implies
    \begin{equation}
        \delta_{t+1} \le \delta_t \left(1 - \frac{1}{T} + \frac{1}{T} O\left(t^{-2}\right)\right) + \tilde{O}\left(t^{-1/3}\right).
    \end{equation}
    This gives us a recurrence relation for $\delta_t$ to upper bound it. For simplicity, let
    \begin{equation}
        a_t = 1 - \frac{1}{T} + \frac{1}{T}O(t^{-2}),
    \end{equation}

    Then by a standard application of the homogeneous recurrence relations formula, we have
    \begin{equation}
        \delta_T = \left( \prod_{t=0}^{T-1} a_t \right) \left( \text{OPT} + \sum_{i=0}^{T-1} \frac{O\left(i^{-1/3}\right)}{\prod_{j=0}^i a_j} \right). \label{eq:homogeneous}
    \end{equation}
    We wish to upper bound $\delta_T$, which means upper and lower bounding $\prod_{j=1}^{i} a_j$. We know $a_j$ is bounded by
    \begin{equation}
        1 - \frac{1}{T} \le a_j \le 1 - \frac{1}{T} + \frac{O(1)}{T}. \label{eq:at_bound} 
    \end{equation}
    Then, by substitution, 
    \begin{equation}
        \left(1 - \frac{1}{T}\right)^T \le \prod_{j=1}^{i} a_j \le \left(1 - \frac{1}{T} + \frac{O(1)}{T} \right)^T.
    \end{equation}
    \begin{equation}
        \exp\left( T \log\left( 1 - \frac{1}{T} \right) \right) \le \prod_{j=0}^{i} a_j \le \exp \left( T \log\left(1 - \frac{1}{T} + \frac{O(1)}{T} \right) \right)
    \end{equation}
    By Taylor expansion of the log terms and simplification, we obtain
    \begin{equation}
        e^{-1-\frac{1}{2i}} \le \prod_{t=0}^{i} a_j \le e^{-1+O(1)}.
    \end{equation}
    Substituting back into Equation~\ref{eq:homogeneous}, we get
    \begin{equation}
        \delta_T \le e^{-1-\frac{1}{2T}} \left(\text{OPT} +  \sum_{m=0}^{T-1} \tilde{O}\left(m^{-1/3}\right) e^{1+\frac{1}{2m}} \right).
    \end{equation}
    Since $e^{1+\frac{1}{2m}} \le O(1)$, 
    \begin{equation}
        \tilde{O}\left(m^{-1/3}\right) O(1) \le \tilde{O}\left(m^{-1/3}\right)
    \end{equation}
    so by integration,
    \begin{align}
        \delta_T &\le e^{-1-\frac{1}{2T}}\left( \text{OPT} + \tilde{O}\left(T^{2/3}\right) \right).
    \end{align}
    Now, since $\delta_T = \text{OPT} - \text{STK}(S_T)$, by substitution and rearranging,
    \begin{equation}
        \mathbb{E}[\text{STK}(S_T)] \ge \left(1 - e^{-1-\frac{1}{2T}}\right) \text{OPT} - \tilde{O}\left(T^{2/3}\right).
    \end{equation}
    
\end{proof}

Notice that as $T \to \infty$, we get
\begin{equation}
    \lim_{T \to \infty} e^{-1-\frac{1}{2T}} = e^{-1}
\end{equation}
Hence, for large $T$, the lower bound is approximately $(1-1/e)\text{OPT} - O(T^{2/3})$. Intuitively, this happens when the greedy rounds become indistinguishable as that of the known distribution setting, with an unavoidable $O(T^{2/3})$ regret incurred from exploration rounds. If this holds, then the algorithm approaches 63\% of the optimal. 

\section{Experiments} \label{sec:experiments}

In this section, we conduct extensive experiments to demonstrate the efficacy of our solution. We answer the following questions. 
\begin{enumerate}
    \item How does \ours perform for intrinsic and extrinsic quality metrics at different time steps?
    \item To what extent are the fallback strategies and histogram maintenance strategies described in \S~\ref{sec:algorithms} effective?
    \item How efficient is \ours compared to simpler baselines?
    \item How well does \ours generalize to different types of data, such as tabular and multimedia data?
    \item How well does \ours generalize to different types of scoring functions, such as regression and fuzzy classification models?
\end{enumerate}

\subsection{Preliminaries} \label{sec:expr-prelim}

Throughout this section, we report sum-of-top-$k$ (STK) as the target intrinsic objective, and Precision@K as the target extrinsic objective. Note that Recall@K is identical to Precision@K as the ground truth solution has $k$ elements. We also host high-resolution plots with additional metrics and configurations in our source code repository.

\subsubsection{Algorithms}

We implement the following algorithms in our standalone system. 
All algorithms sample without replacement. 

\begin{enumerate}
    \item \ours: Algorithm~\ref{alg:cont-epsgreedy} combined with the index of \S~\ref{sec:index}. We use $B = 8$, $\alpha = 0.1$, $\beta = 1.1$, and $F = 0.01$ by default.
    \item \ucb: A standard upper confidence bound (UCB) bandit algorithm combined with the index of \S~\ref{sec:index}. We set the exploration parameter as 1.0 and initialize the mean using query-specific prior knowledge. 
    \item \explore: A bandit which chooses a uniformly random non-empty child in each layer of the index of \S~\ref{sec:index}.  
    \item \sample: Uniform sampling over the entire search domain, implemented via pre-shuffling of the data, then performing a sequential scan. \sample represents the average case result of \scan, as there is no additional run-time overhead. 
    \item \scanbest, \scanworst: \scan over the domain where the elements are sorted in the best-case or worst-case order. This is meant to demonstrate theoretical limits. 
    \item \sortedscan: \scan over an in-memory sorted index built on a new column that contains pre-computed UDF function values. \sortedscan skips scoring function evaluation and priority queue maintenance.
\end{enumerate}

\subsubsection{Datasets}

\begin{enumerate}
    \item \textbf{Synthetic data:} We randomly generate $L$ normal distributions with $\mu \in [0,20]$ and $\sigma \in (0, 5]$. We then draw a fixed number of samples from each distribution, which then serves as the leaf clusters of the index. We build the dendrogram over the means of each cluster. There are 20 clusters and 2,500 samples per cluster.  
    \item \textbf{Tabular data:} \sloppy{We use the US Used Cars dataset~\cite{usedcars}, collected from an online reseller~\cite{cargurus}. We clean the data by projecting boolean and numeric columns, turning each column numeric, imputing missing values with column mean, then normalizing input features. 
    The cleaned dataset has 11 columns: three boolean columns (\texttt{frame damaged}, \texttt{has accidents}, \texttt{is new}), six numeric columns (\texttt{daysonmarket}, \texttt{height}, \texttt{horsepower}, \texttt{length}, \texttt{mileage}, \texttt{seller rating}), one target column (\texttt{price}), and one key column (\texttt{listing id}).
    The target feature (price) is used for model training but excluded from indexing and querying. We build an index over a subset of $100,000$ rows with 500 leaf clusters.} 
    
    \item \textbf{Multimedia data:} We use the ImageNet~\cite{deng2009imagenet} dataset. 
    We build the index and perform queries over a random subset of $320,000$ images. We scale each image to a 16x16x3 tensor, including the color channels, and flatten it to a vector. We then apply $k$-means clustering over a subsample of 100,000 images with 25 clusters. Finally, we assign all images to their closest cluster and build a dendrogram. 
\end{enumerate}

\subsubsection{Scoring Functions}

\begin{enumerate}
    \item \textbf{Synthetic data:} The scoring function for synthetic data is the simple ReLU function, $f(x) = \max(0, x)$, to ensure non-negativity. 
    \item \textbf{Tabular data:} We train a regression model to predict a listing's price, using XGBoost~\cite{chen2016xgboost} with default parameters, on a subset of 1 million rows. The train split is disjoint from the split used for indexing and query evaluation. We use a batch size of 1 on CPU for inference. 
    \item \textbf{Multimedia data:} We use a pre-trained ResNeXT-64 model's softmax layer to obtain its confidence that an image belongs to a particular label. Three  target labels were chosen randomly. We use a batch size of 400 on GPU for inference. 
\end{enumerate}

\subsubsection{Hardware Details} 
We ran all experiments on a dedicated server with a 13th Gen Intel i5-13500 (20 cores @ 4.8GHz) CPU, an NVIDIA RTX™ 4000 SFF Ada Generation GPU, and 64GB RAM running Ubuntu 22.04.

\subsection{Synthetic Data}

\begin{figure*}[htb]
    \centering
    \subfloat[STK]{
        \includegraphics[width=0.3\textwidth]{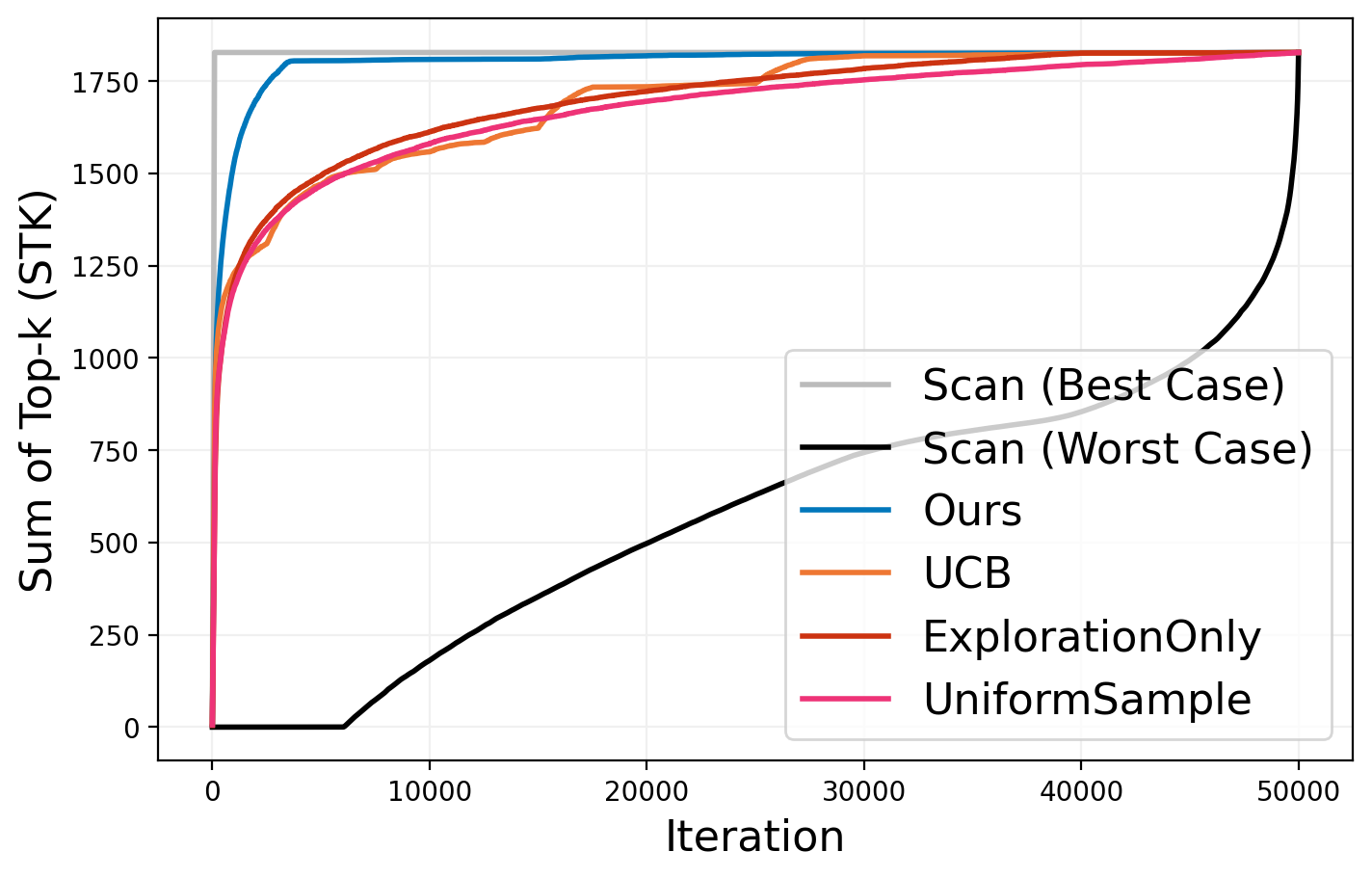}
        \label{fig:synthetic-stk}
    }
    \hfill
    \subfloat[Precision@K]{
        \includegraphics[width=0.3\textwidth]{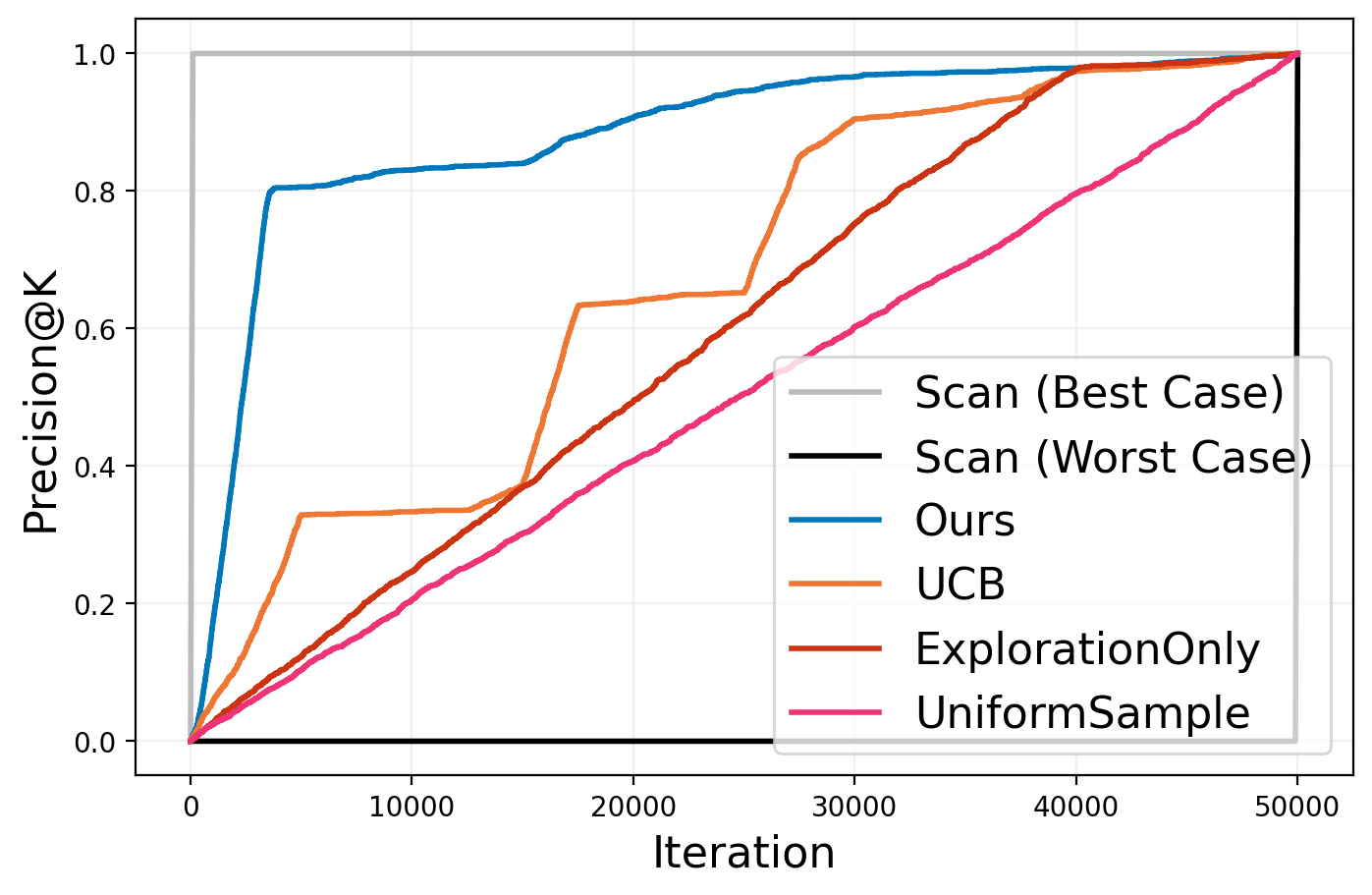}
        \label{fig:synthetic-precision}
    }
    \hfill
    \subfloat[Ablation study]{
        \includegraphics[width=0.3\textwidth]{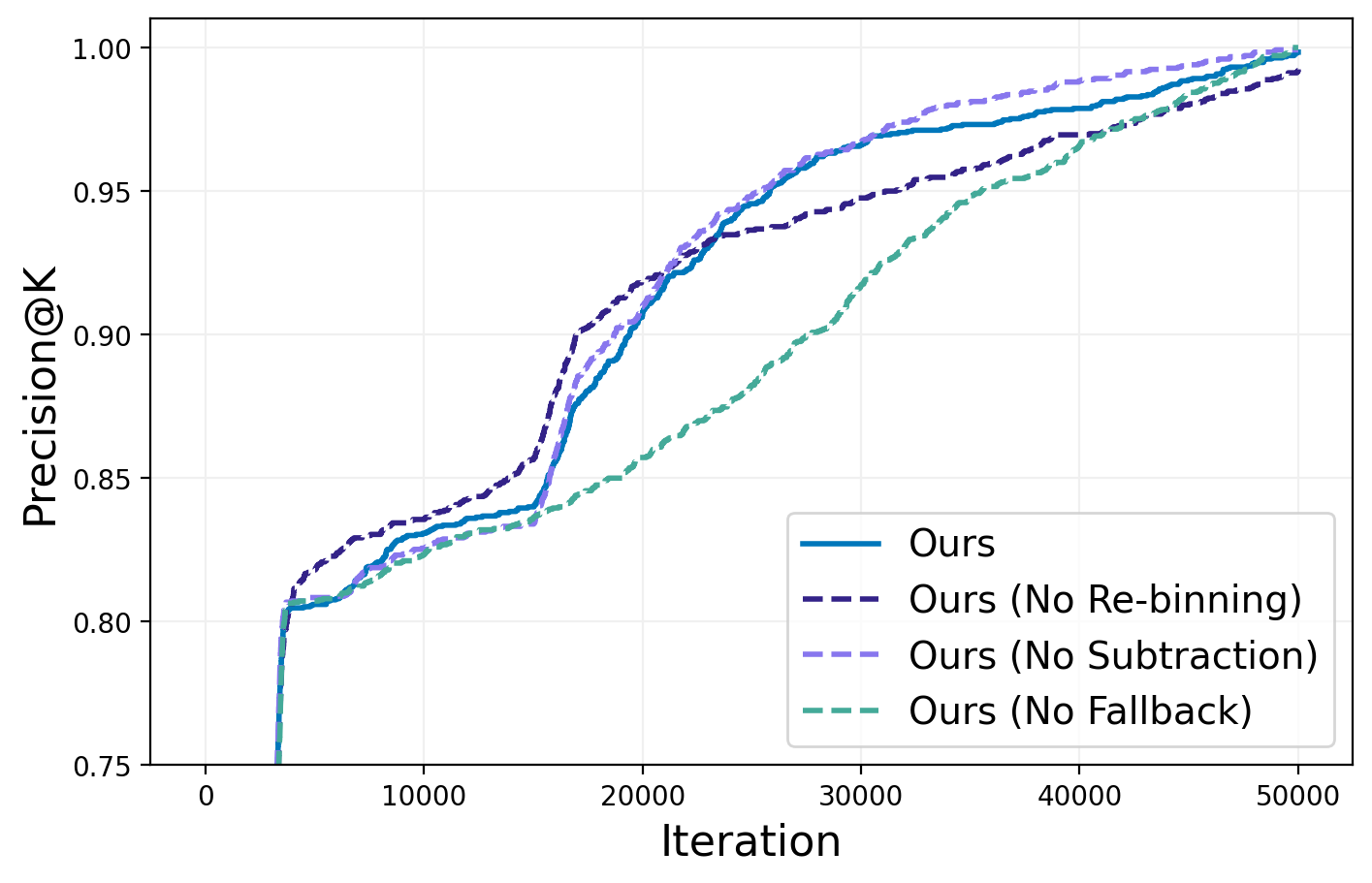}
        \label{fig:synthetic-ablation}
    }
    \caption{Selecting $100$ highest numbers from synthetic normally distributed data ($n = 50,000$) with 20 clusters. Averaged over 25 runs. (a-b) STK, Precision@K vs time. (c) Ablation study. Note the $y$-axis cutoff.}
    \label{fig:synthetic}
\end{figure*}

We run our algorithm and baseline algorithms on synthetic data, with $k = 100$ and $T = 50,000$. We also conduct an ablation study. 
We do not report latency for this dataset, since the number of iterations is an indicator of latency for synthetic data.

Figure~\ref{fig:synthetic} shows the result of this experiment. There are three main takeaways. First, \ours out-performs baseline algorithms in terms of both STK (Figure~\ref{fig:synthetic-stk}) and Precision@K (Figure~\ref{fig:synthetic-precision}). 
\ours reaches near-optimal STK rapidly. 
However, near-optimal Precision@K requires a much larger number of iterations. 
Second, the curve for \ucb is piecewise-continuous, since it chooses sub-optimal intermediate nodes after a good child is exhausted. 
Third, turning off various features does not significantly impact performance. That said, turning off fallback has the most negative impact.

\subsection{Tabular Regression}

\begin{figure*}[htb]
    \centering
    
    \subfloat[STK]{
        \raisebox{4pt}{
        \includegraphics[width=0.29\textwidth]{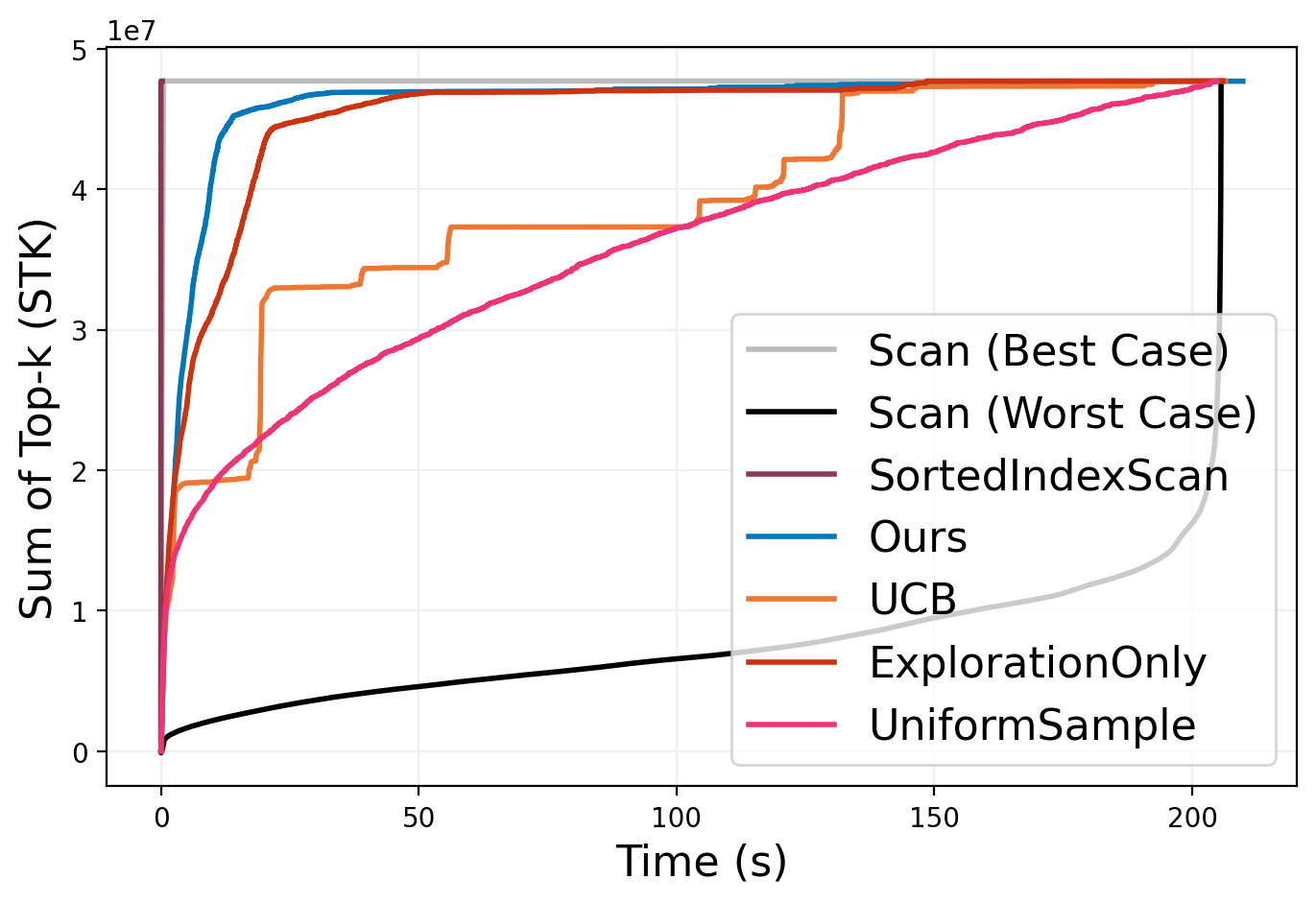}}
        \label{fig:usedcars-stk}
    }
    \hfill
    \subfloat[Precision@K]{
        \raisebox{4pt}{
        \includegraphics[width=0.29\textwidth]{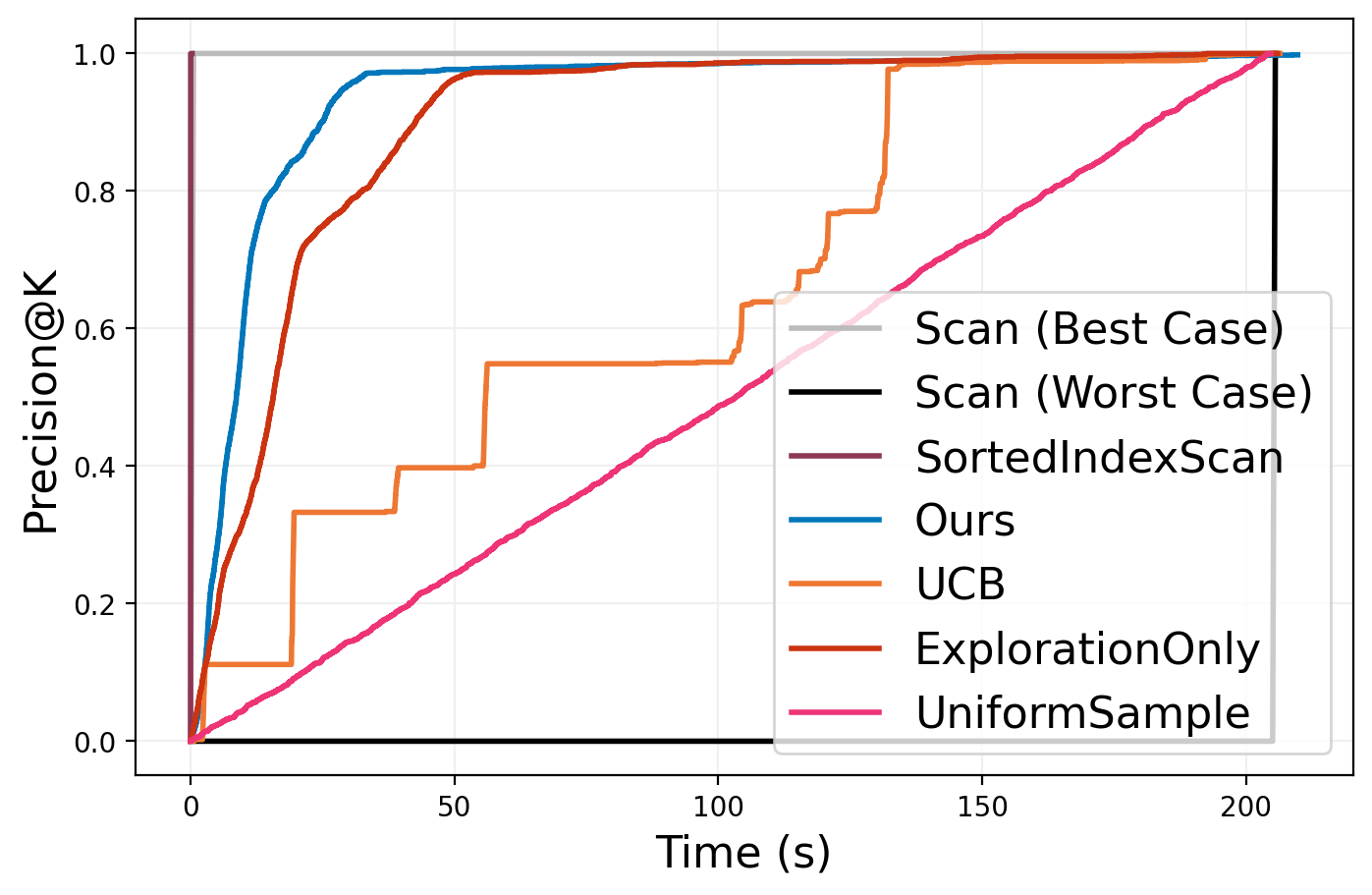}}
        \label{fig:usedcars-precision}
    }
    \hfill
    \subfloat[End-to-end latency]{
        \includegraphics[width=0.29\textwidth]{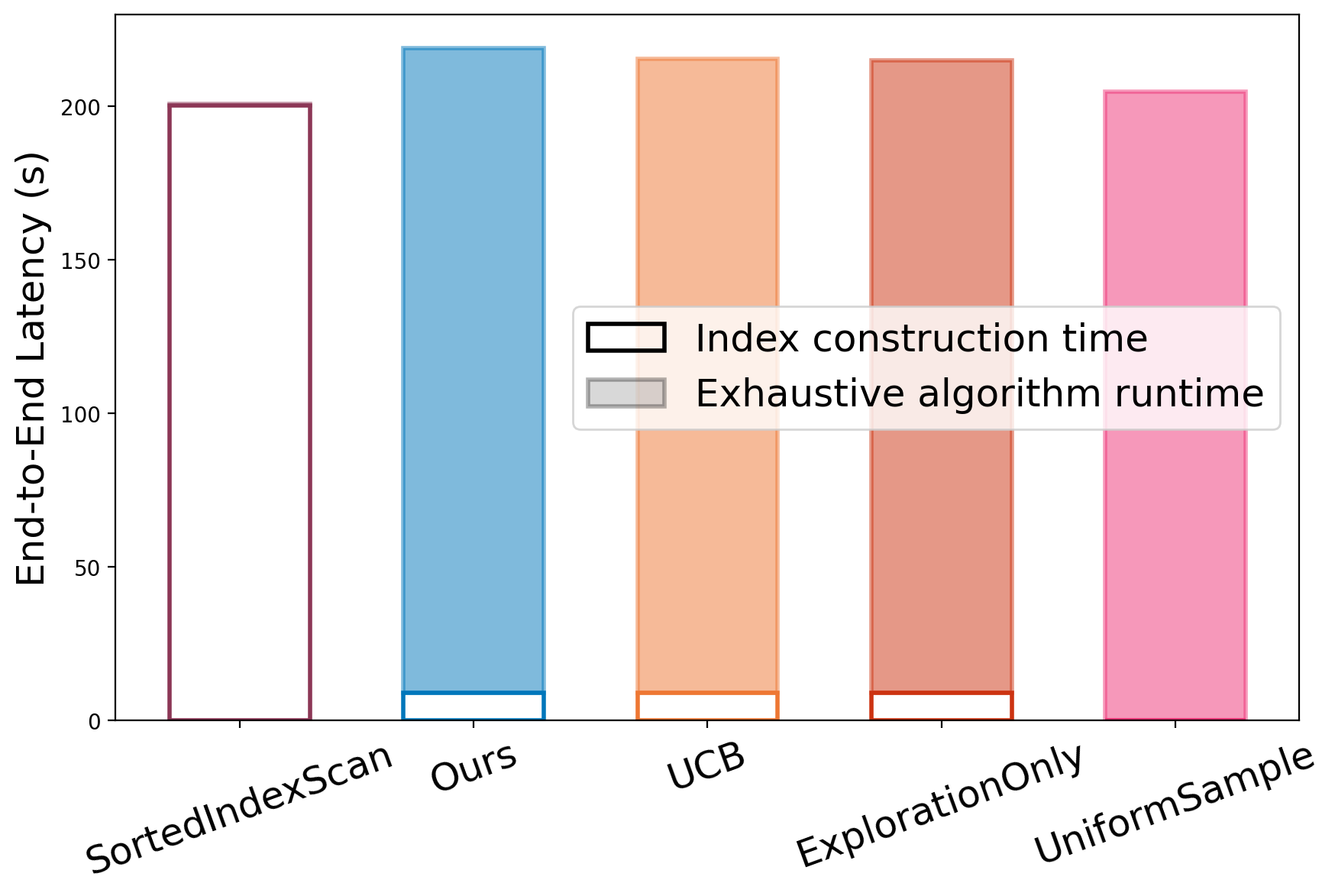}
        \label{fig:usedcars-latency-total}
    }
    \caption{Selecting 250 highest valued car listings form the UsedCars dataset ($n = 100,000$) where the valuation is given by an XGBoost model. Averaged over 10 runs. (a-b) STK, Precision@K vs time. (c) End-to-end latency, which includes the index building time and the time to run the algorithm exhaustively over the entire dataset.}
    \label{fig:usedcars1}
\end{figure*}

\begin{figure*}[htb]
    \centering
    
    \subfloat[Ablation study]{
        \raisebox{10pt}{\includegraphics[width=0.3\textwidth]{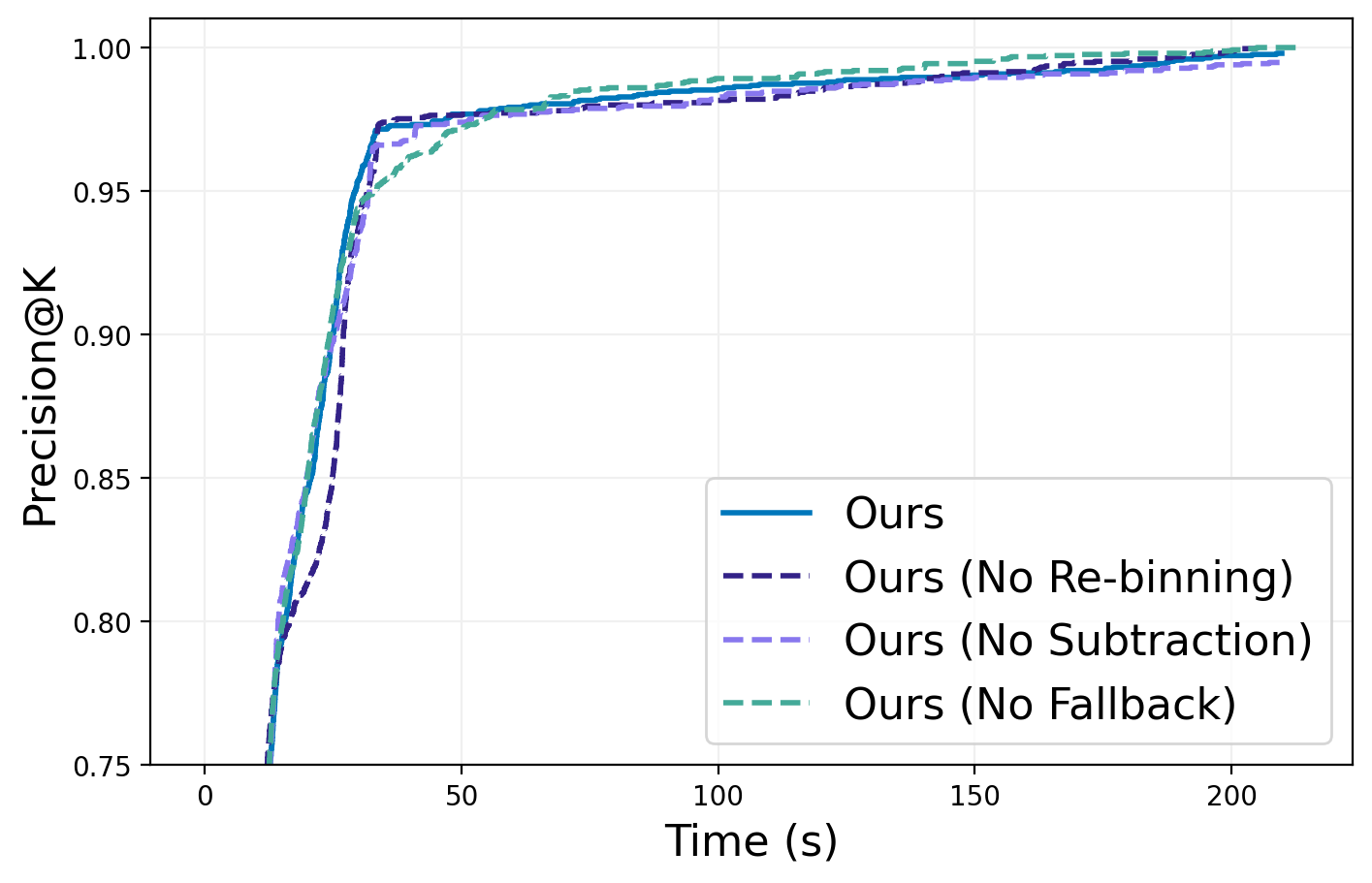}}
        \label{fig:usedcars-ablation}
    }
    \hfill
    \subfloat[Overhead per iteration]{
        \raisebox{0pt}{\includegraphics[width=0.3\textwidth]{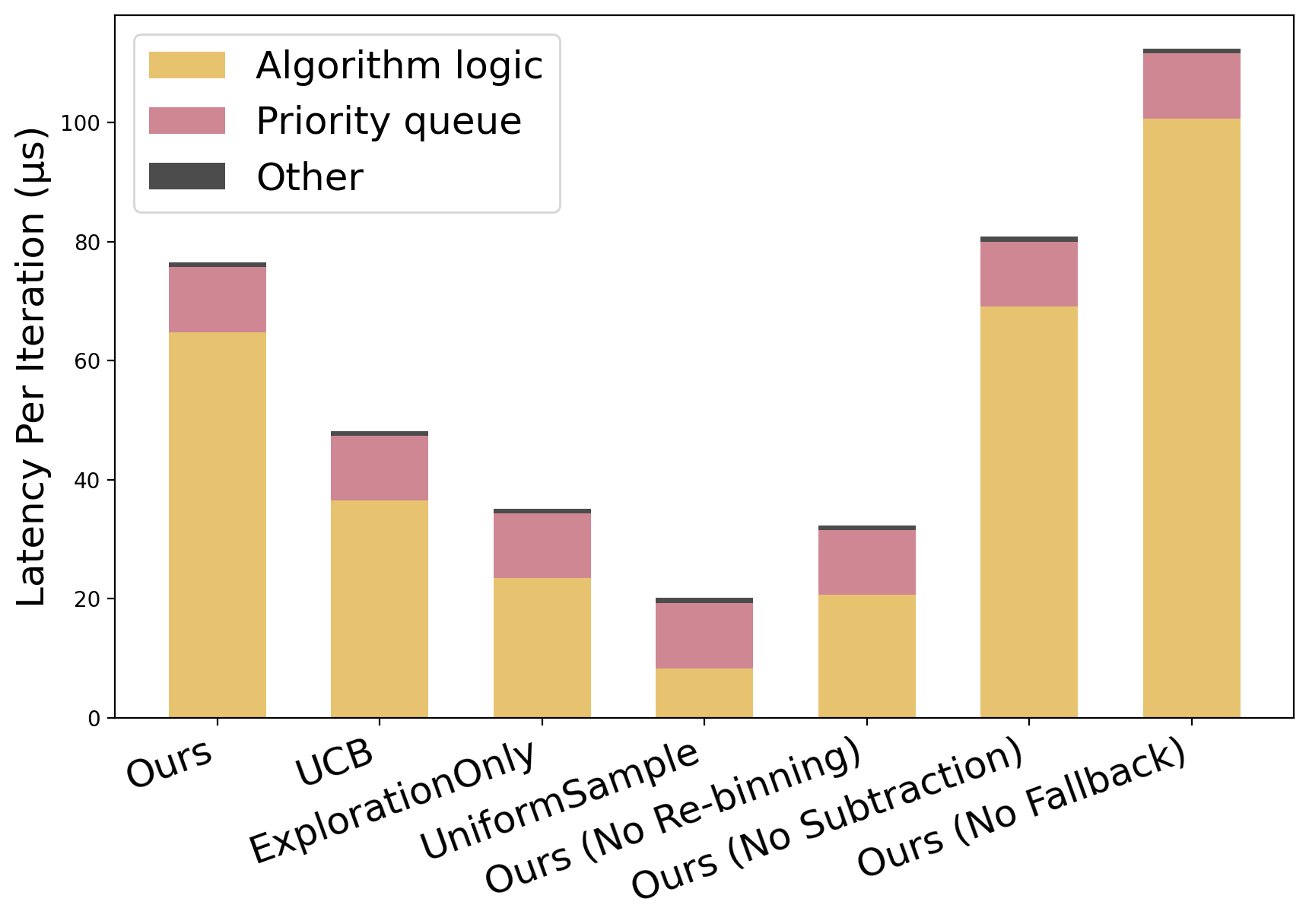}}
        \label{fig:usedcars-latency-iter}
    }
    \hfill
    \subfloat[Parameter study]{
        \raisebox{10pt}{\includegraphics[width=0.3\textwidth]{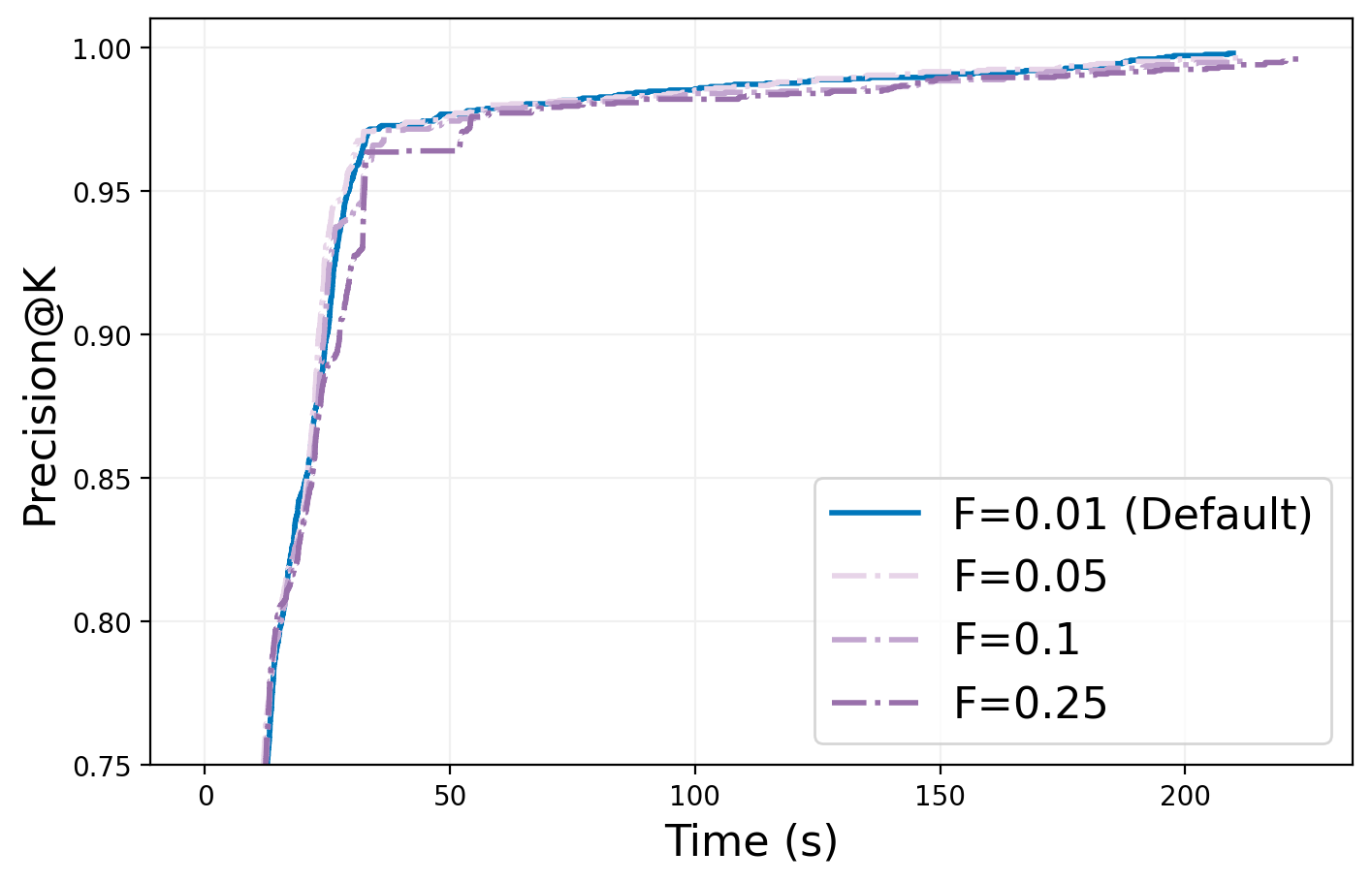}}
        \label{fig:usedcars-params}
    }
    \caption{Additional UsedCars results. (a) Ablation study. Note the $y$-axis cutoff. (b) Overhead incurred by different algorithms per iteration. Excludes scoring function latency (2ms/iter). (c) Parameter study. Variants have different fallback condition checking frequencies ($F$). Note the $y$-axis cutoff.}
    \label{fig:usedcars2}
\end{figure*}

To test the applicability of \ours on tabular regression, we evaluate the algorithms on the task of selecting cars with the highest predicted valuation. This task also demonstrates the flexibility of the histogram re-binning strategy, as the maximum predicted price is unknown and unbounded. We set $k = 250$ and $T = 100,000$. 
We also report a \sortedscan baseline, end-to-end and per-iteration latency analysis, and an ablation study.

\textbf{Comparison to \sortedscan.} As shown in Figure~\ref{fig:usedcars-stk} and \ref{fig:usedcars-precision}, \sortedscan is very fast in query time, since it skips scoring function evaluation. However, as we observe in Figure~\ref{fig:usedcars-latency-total}, \sortedscan has a much higher index construction time than \ours. If the user only needs an approximate result, then \ours is much faster.

\textbf{Comparison to other baselines.} Figure~\ref{fig:usedcars-stk} and \ref{fig:usedcars-precision} show the quality of results for \ours and all baselines. 
\ours significantly out-performs baselines and obtains near-optimal STK and Precision@K in a small number of iterations, as it identifies a few clusters that contain most of the exact solution. 
\explore performs well, sometimes eclipsing \ours. 
We believe this is caused by two reasons. 
A few shallow leaf nodes that contain a large proportion of the ground truth solution, which \explore is skewed towards. 
Furthermore, if the uniform value assumption does not hold, then \ours can fail to model the exact distributions. 
\ucb significantly under-performs. 
Most likely, maximizing expected reward in each iteration causes \ucb to select arms with high mean and low variance, which does not improve the running solution.

\textbf{Parameter and Ablation Study.} Figure~\ref{fig:usedcars-ablation} shows the result of an ablation study. 
All variants perform similarly, though with minor performance degradations if some features are turned off. 
Figure~\ref{fig:usedcars-params} shows a comparison of different fallback frequencies. Reducing the frequency slightly diminishes performance, but has minor impact.

Figure~\ref{fig:usedcars-latency-iter} shows the overhead of different algorithms. 
While \ours has high overhead, scoring function latency (2ms) is 18-25x longer. 
While enabling fallback incurs additional costs, it reduces average overhead over the entire query. 
Skipping re-binning decreases overhead, whereas skipping subtraction has little effect.

\subsection{Image Fuzzy Classification}

\begin{figure*}[htb]
    \centering
    \subfloat[Beacon, lighthouse (STK)]{
        \includegraphics[width=0.3\textwidth]{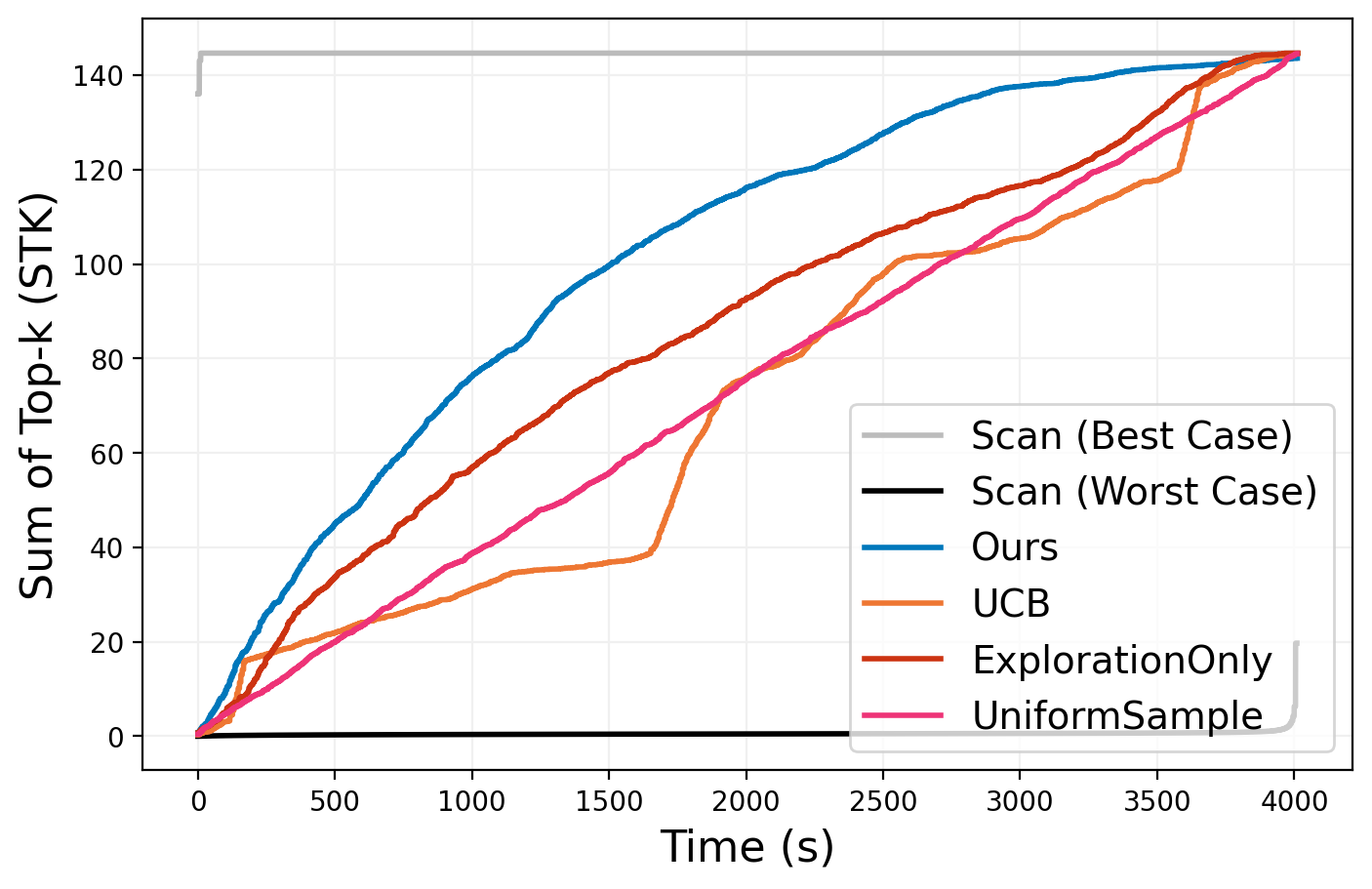}
        \label{fig:classify-1-stk}
    }
    \hfill
    \subfloat[Hand-held computer (STK)]{
        \includegraphics[width=0.3\textwidth]{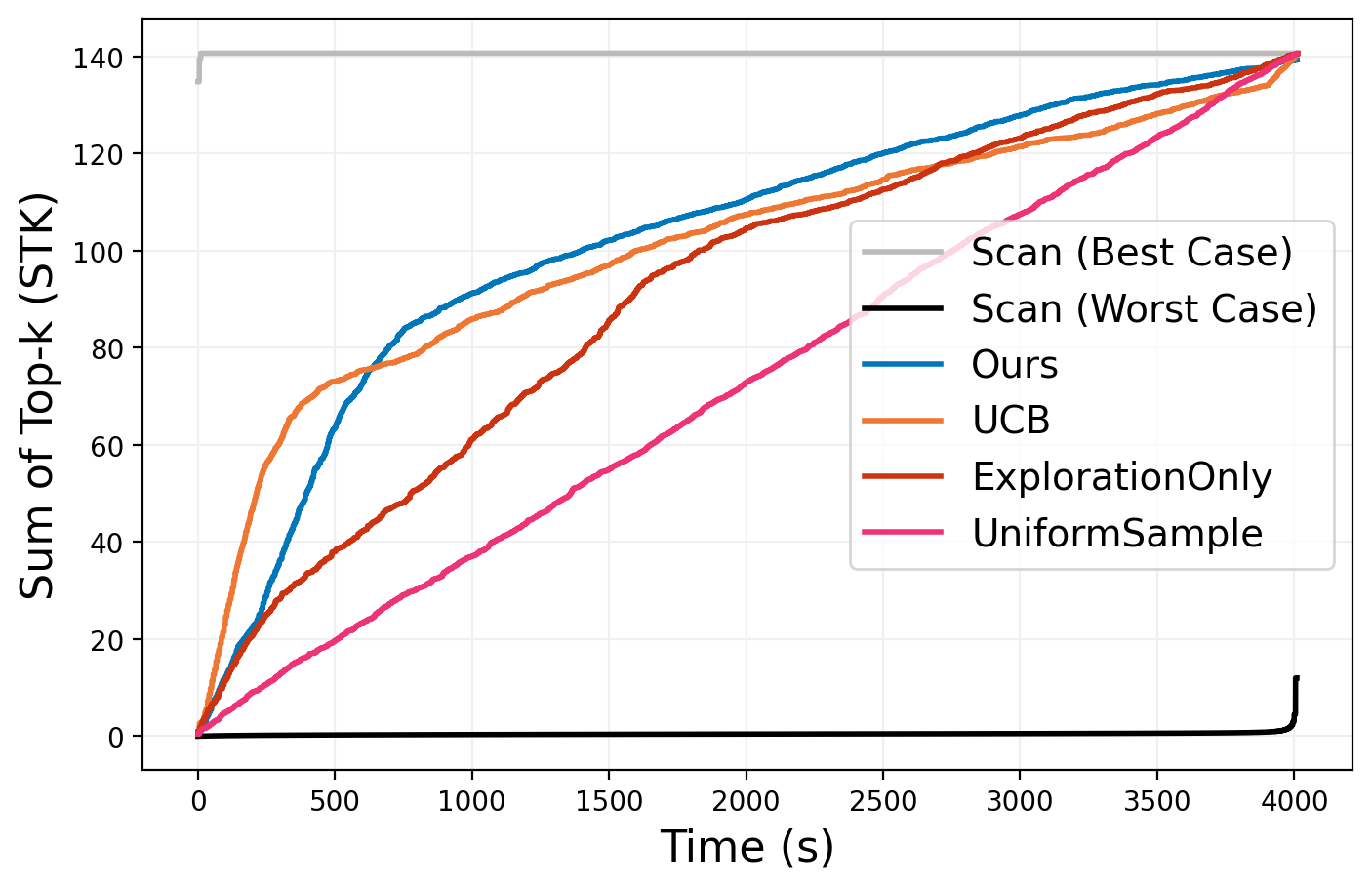}
        \label{fig:classify-2-stk}
    }
    \hfill
    \subfloat[Washing machine (STK)]{
        \includegraphics[width=0.3\textwidth]{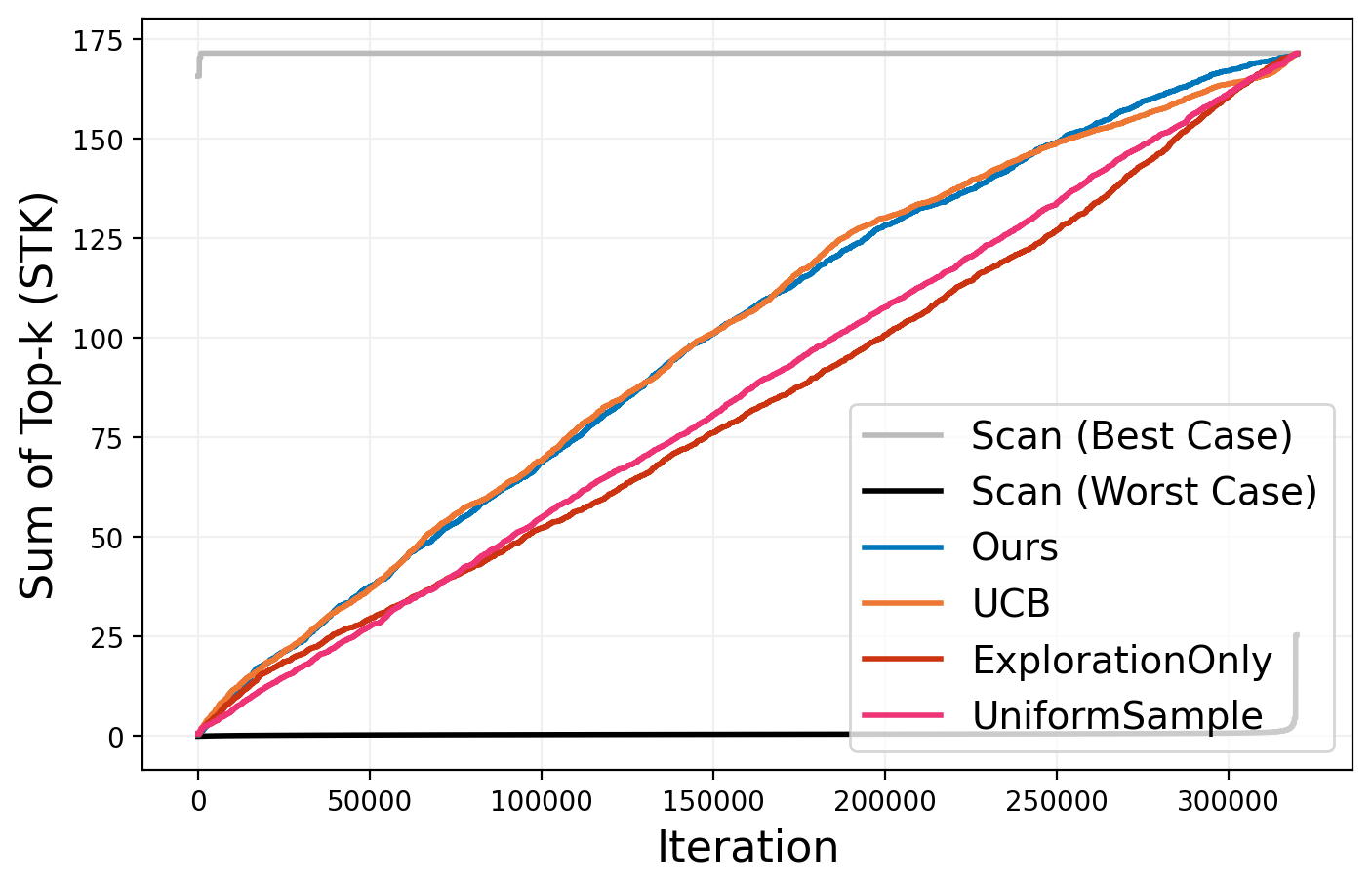}
        \label{fig:classify-3-stk}
    }
    \\
    \subfloat[Beacon, lighthouse (Precision@K)]{
        \includegraphics[width=0.3\textwidth]{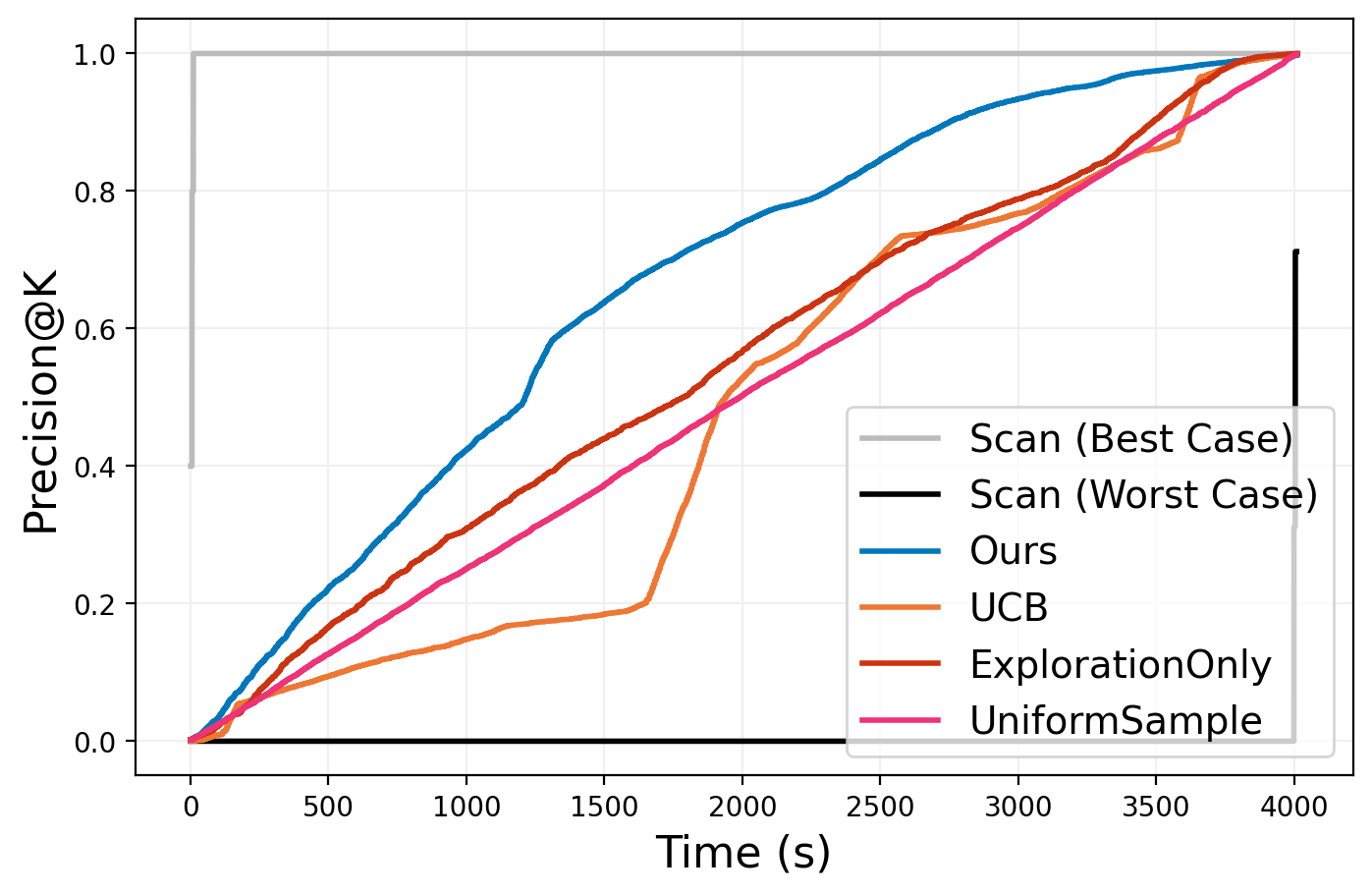}
        \label{fig:classify-1-precision}
    }
    \hfill
    \subfloat[Hand-held computer (Precision@K)]{
        \includegraphics[width=0.3\textwidth]{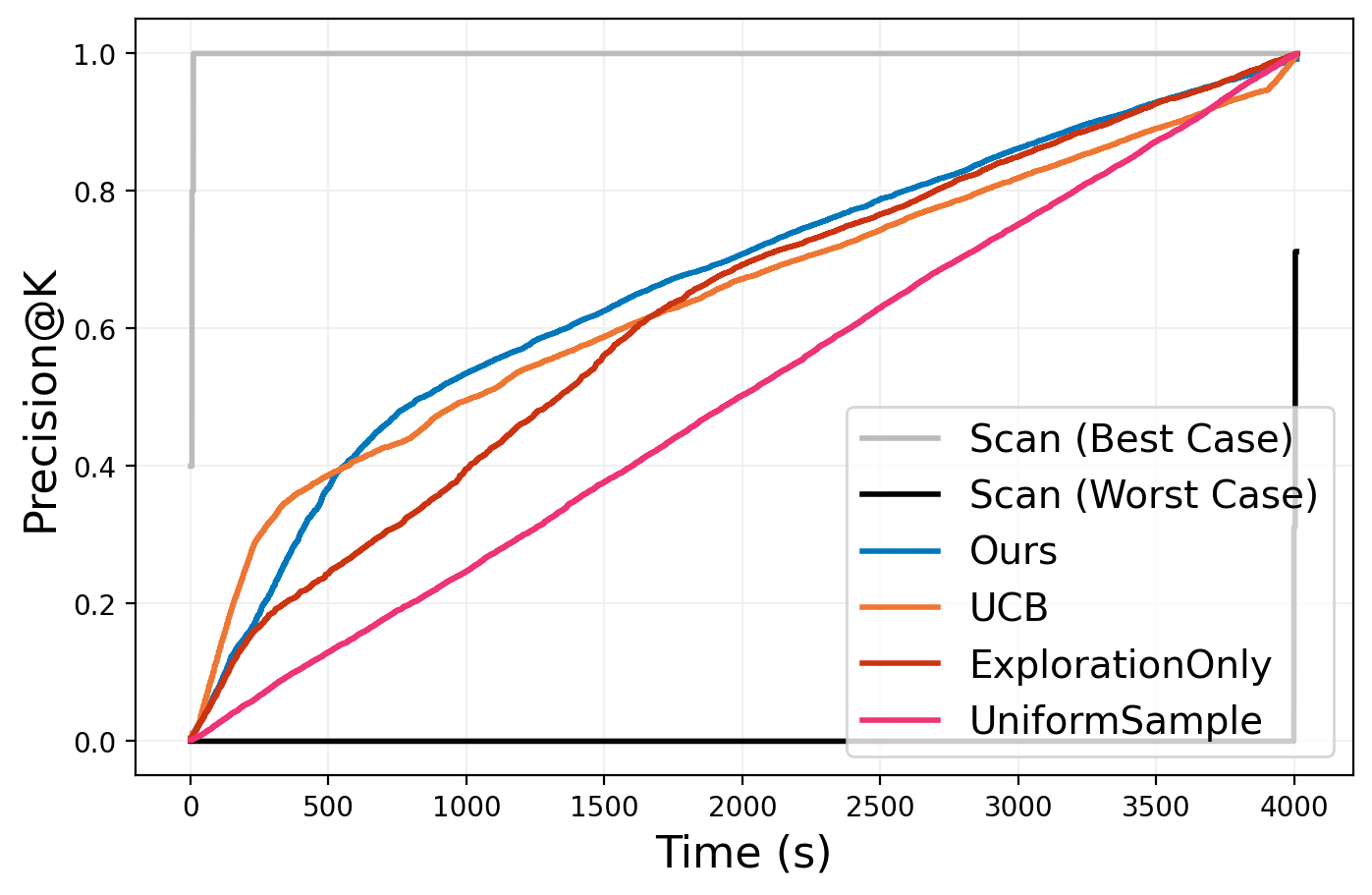}
        \label{fig:classify-2-precision}
    }
    \hfill
    \subfloat[Washing machine (Precision@K)]{
        \includegraphics[width=0.3\textwidth]{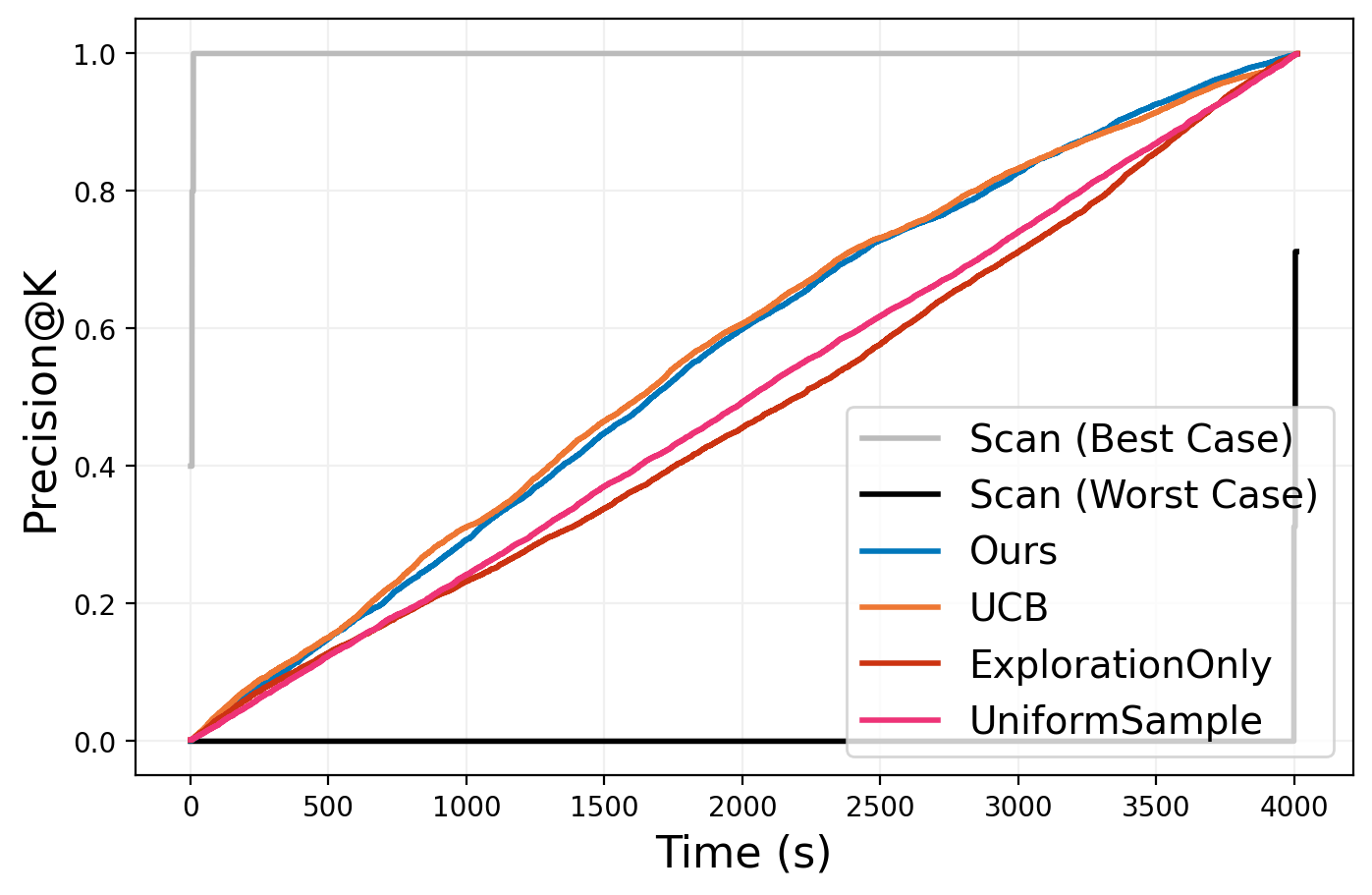}
        \label{fig:classify-3-precision}
    }
    \caption{Selecting top 1000 images most confidently classified as a class label by a pre-trained ResNeXT model, over a subset of ImageNet ($n = 320,000$). Averaged over 10 runs. (a-c) STK vs time. (d-f) Precision@K vs time.}
    \label{fig:classify}
\end{figure*}

\begin{figure*}[htb]
    \centering
    \subfloat[Batch size vs. latency, memory]{
        \includegraphics[width=0.3\textwidth]{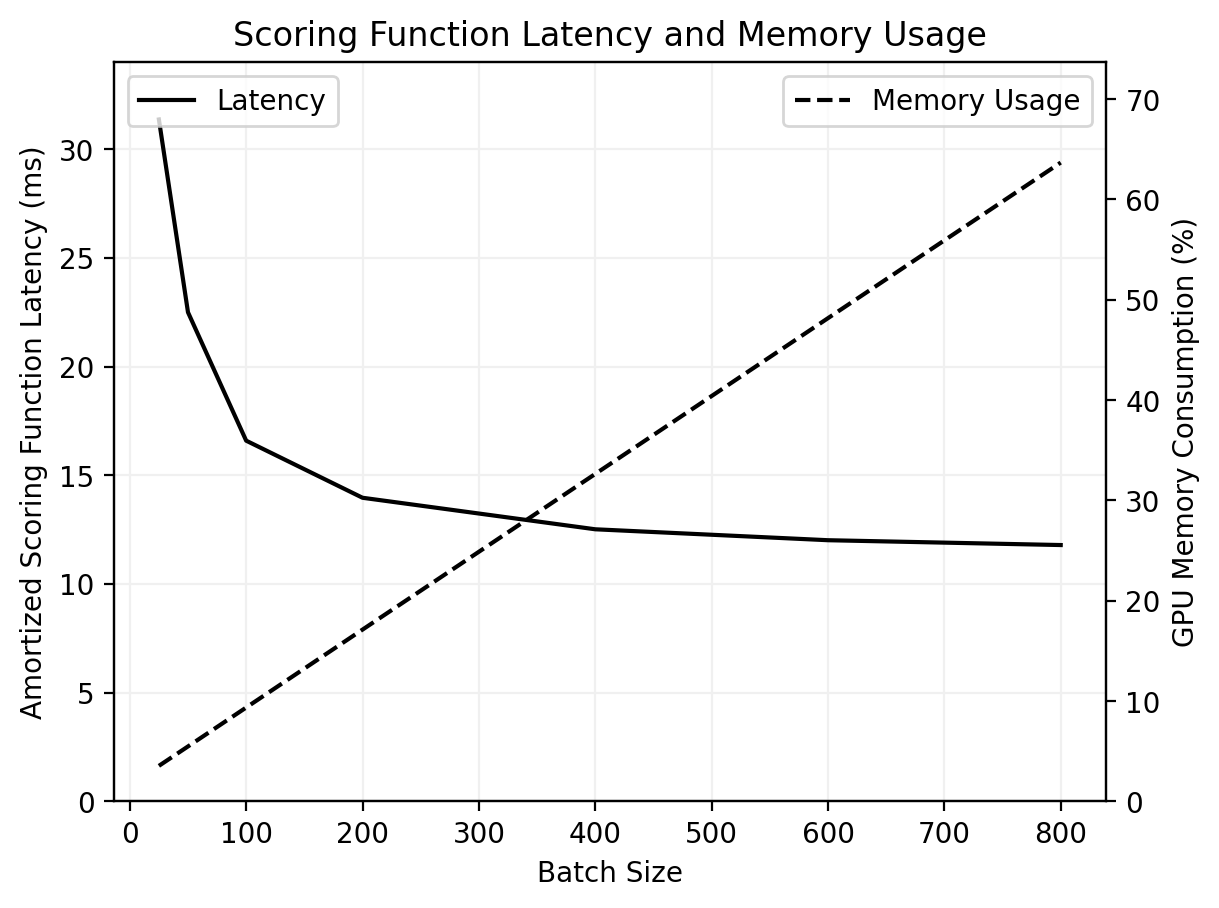}
        \label{fig:imagenet-batchsize}
    }
    \hfill
    \subfloat[End-to-end latency]{
        \includegraphics[width=0.3\textwidth]{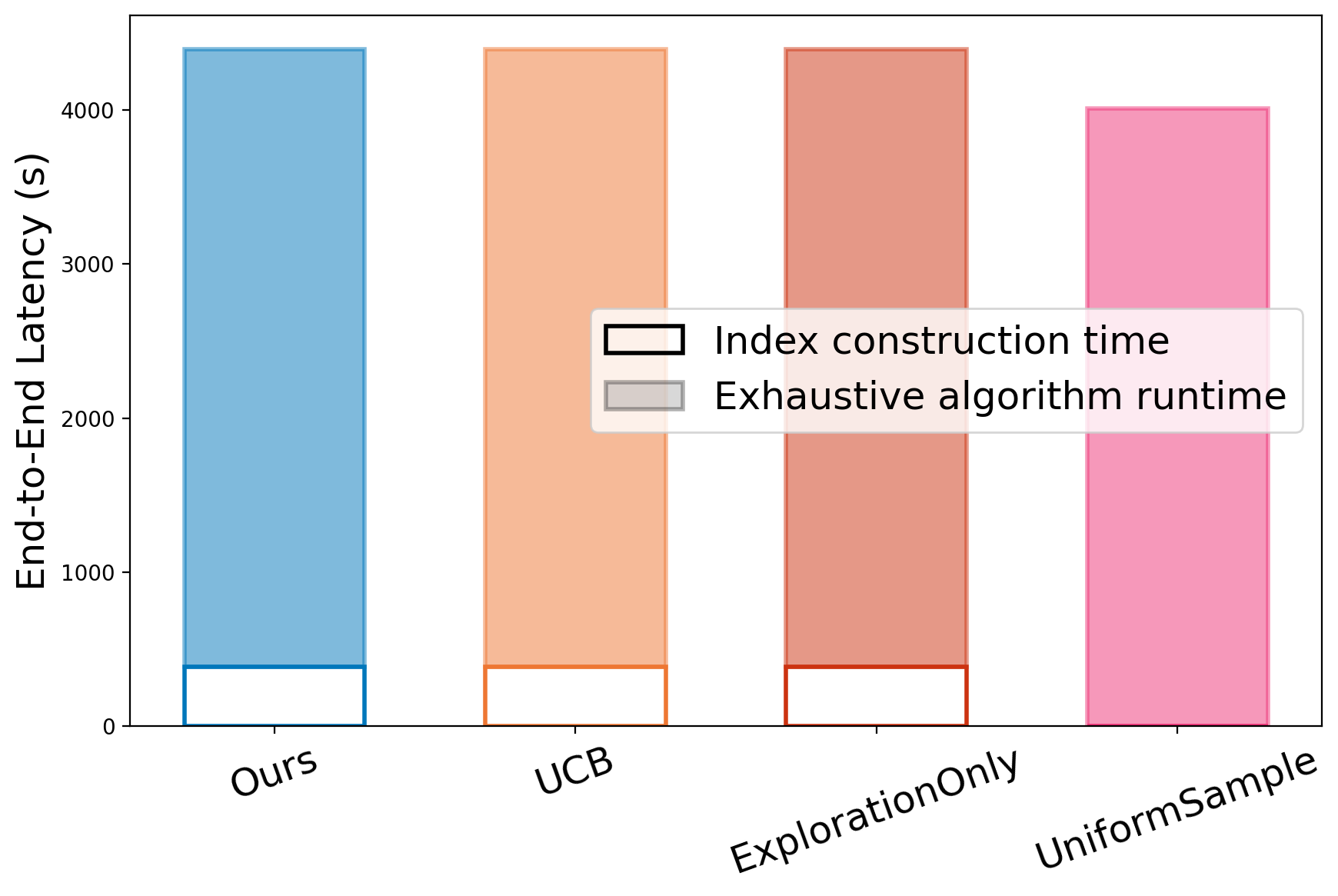}
        \label{fig:imagenet-latency-total}
    }
    \hfill
    \subfloat[Overhead per iteration]{
        \includegraphics[width=0.3\textwidth]{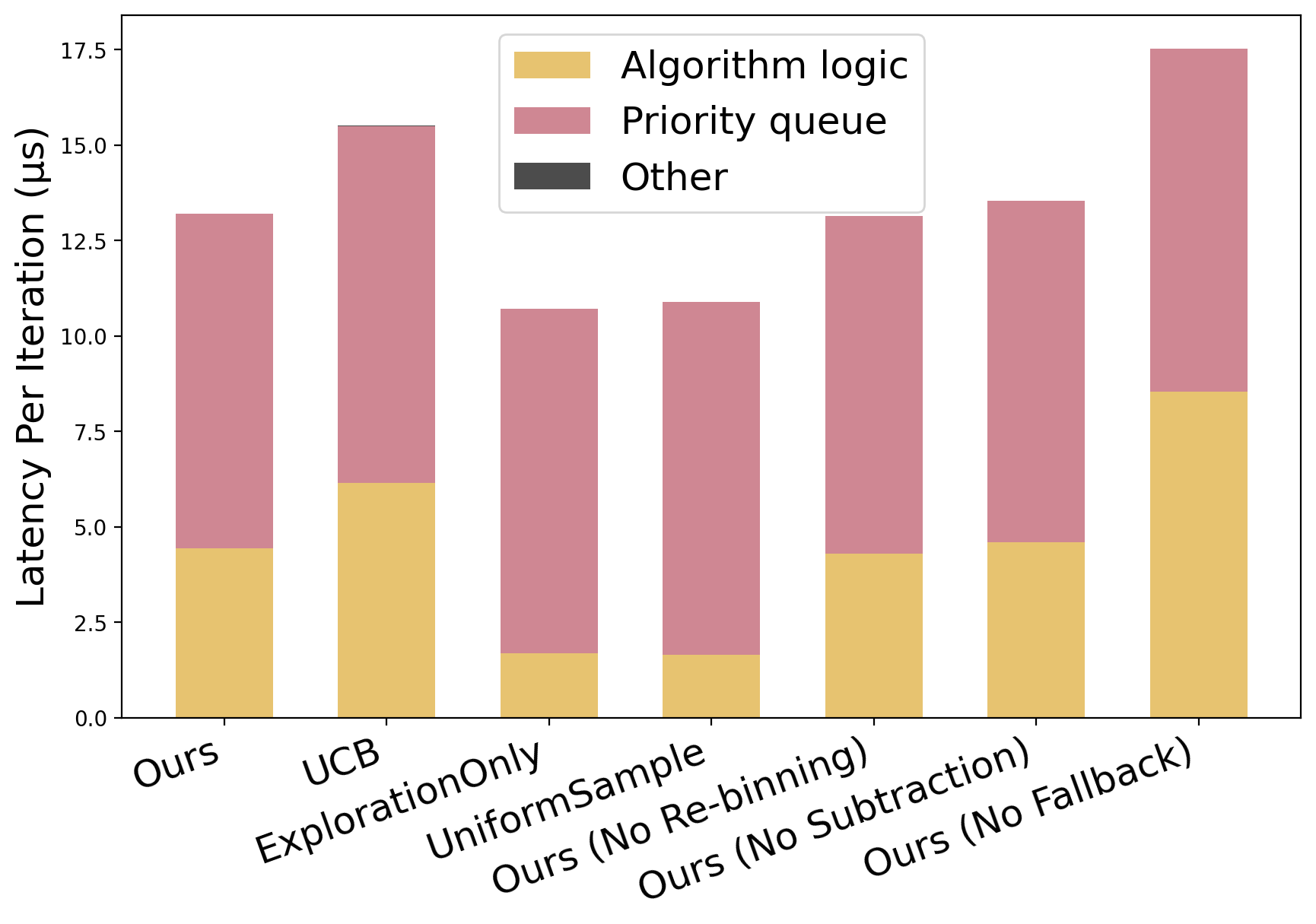}
        \label{fig:imagenet-latency-iters}
    }
    \caption{ImageNet latency results. (a) Scoring function latency, GPU memory consumption vs batch size. (b) End-to-end latency, which includes index buliding time and the time to run the algorithm exhaustively over the entire dataset. (c) Overhead incurred by different algorithms and ablation variants per iteration. Excludes scoring function latency (13ms/iter).}
    \label{fig:imagenet-latency}
\end{figure*}

\begin{figure*}[htb]
    \centering
    \subfloat[Parameter study (beacon, lighthouse)]{
        \includegraphics[width=0.32\textwidth]{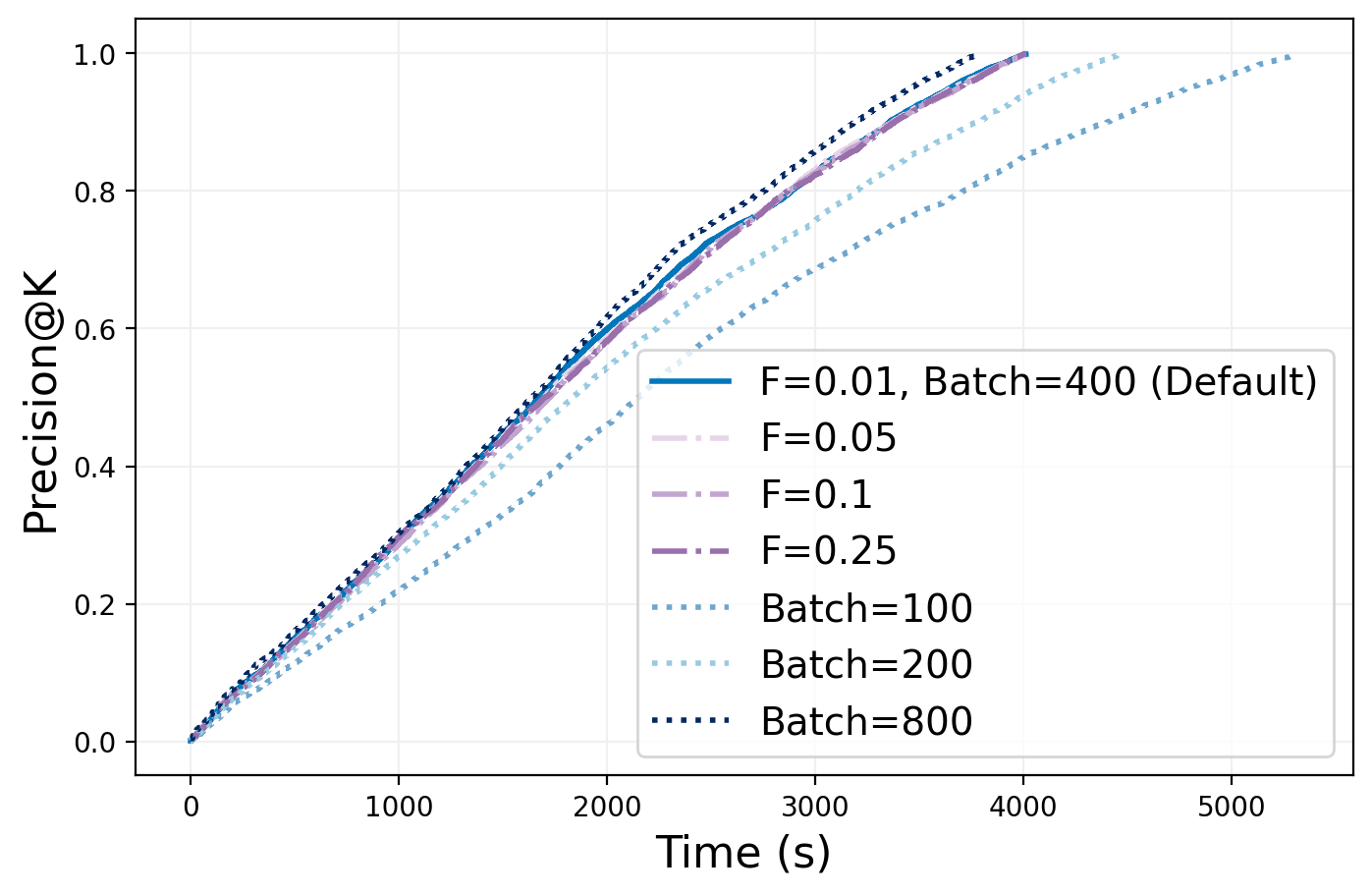}
        \label{fig:imagenet-params-1}
    }
    \hfill
    \subfloat[Parameter study (hand-held computer)]{
        \includegraphics[width=0.32\textwidth]{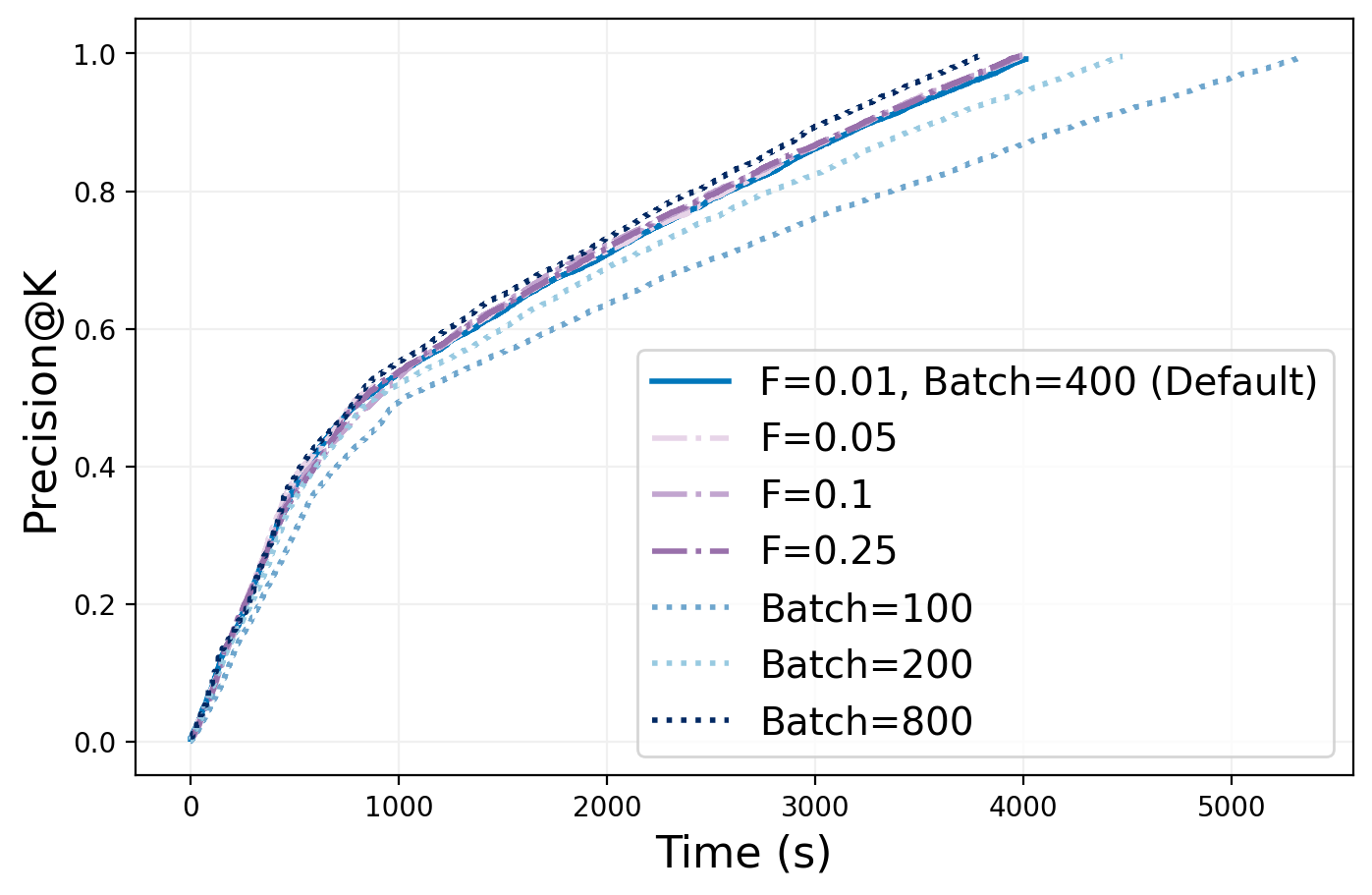}
        \label{fig:imagenet-params-2}
    }
    \hfill
    \subfloat[Parameter study (washing machine)]{
        \includegraphics[width=0.32\textwidth]{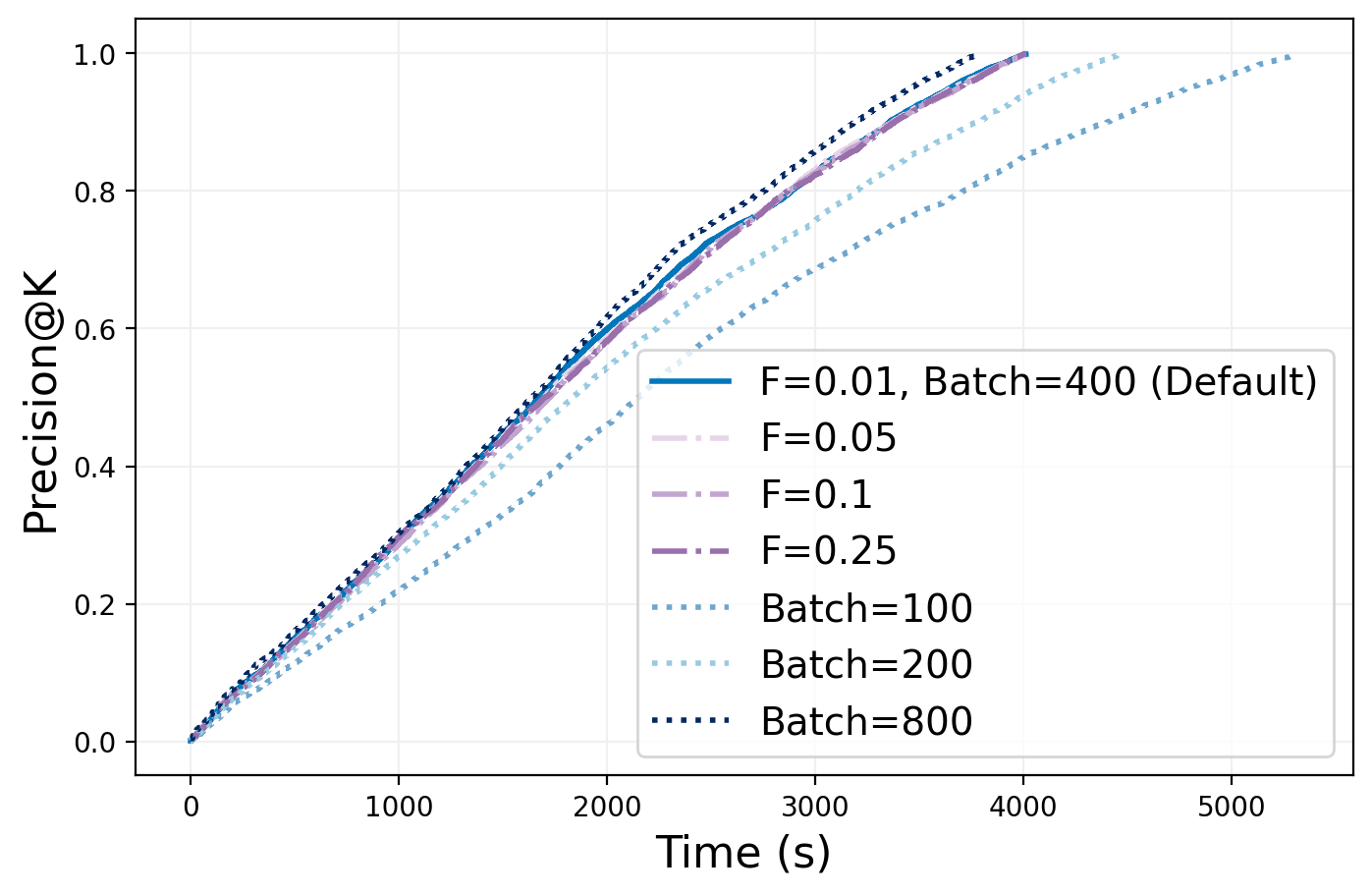}
        \label{fig:imagenet-params-3}
    }
    \caption{ImageNet parameter study. Variants have different fallback condition checking frequency ($F$) and batch sizes.}
    \label{fig:imagenet-params}
\end{figure*}

We evaluate algorithms for selecting images most confidently classified as some label. 
This task generalizes the binary classification task in He et al. \cite{he2020method} to fuzzy classification. 
We let $k = 1000$ and $T = 320,000$. 
Three random labels were chosen. 
We also report a detailed latency analysis and a parameter study.

\textbf{Comparison to baselines.} Figure~\ref{fig:classify} shows the main result of this experiment. 
There are three main takeaways. 
First, \ours almost always out-performs baseline algorithms. 
Second, \ucb is sometimes on par with (Figure~\ref{fig:classify-3-stk}), or out-performs \ours (Figure~\ref{fig:classify-2-stk}). 
This might be caused by the reduced statistical efficiency of the $\varepsilon$-greedy strategy with a large batch size, or if the edge case when \ucb is suboptimal does not occur. 
Then, the ability to simultaneously explore and exploit using \ucb might be advantageous, even if \ucb has no approximation guarantee. 
Third, the amount of advantage that \ours has over the baseline algorithms varies heavily for different labels.
Images that belong to some labels have a consistent visual pattern, whereas others may not.

\textbf{Latency analysis.} Figure~\ref{fig:imagenet-latency} analyzes the latency of \ours and baselines. 
Figure~\ref{fig:imagenet-batchsize} shows that scoring function latency decreases as batch size increases, but with diminishing returns, as the model becomes compute bound. 
GPU memory capacity is not a bottleneck. 
Figure~\ref{fig:imagenet-latency-total} shows the end-to-end latency of \ours and baselines. 
The index building cost is recouped after one or two queries if the user only needs an approximate solution. 
Figure~\ref{fig:imagenet-latency-iters} shows the overhead of various algorithms. 
The scoring function latency of 13ms per iteration is 70x longer than the highest algorithm overhead of 18$\mu$s. 
Enabling the fallback strategy reduces overall overhead. 
Skipping re-binning or subtraction has negligible impact.

\textbf{Parameter study}. Figure~\ref{fig:imagenet-params} shows the impact of various batch sizes and fallback frequencies. 
Batch size of 800 slightly out-performs 400, which means the lower scoring function latency more than compensates for the loss in statistical efficiency. 
This might be because an average leaf cluster contains over 10,000 images. 
Conversely, decreasing the batch size degrades performance. 
Modifying fallback checking frequency has negligible impact. 

\subsection{Summary \& Recommendations}

\ours almost always out-performs the baseline algorithms. 
While there are edge cases where \ucb or \explore performs well, there are also scenarios where they perform poorly. 
If scoring function latency is low, then the overhead of \ours might be non-negligible. 
If so, turn off the re-binning feature for lower overhead. 
Check the fallback condition frequently, as the fallback strategy is highly effective. 
If applicable, use large batch sizes to reduce latency, and set larger leaf cluster sizes to compensate. 

\section{Related Work} \label{sec:related}

\textbf{Statistical Methods for UDF Optimization} Our paper is most similar to He et al. \cite{he2020method}, which optimizes opaque filter queries with cardinality constraints. We borrow their indexing scheme and generalize filter functions to scoring functions. In doing so, we replace the UCB bandit with a histogram-based $\varepsilon$-greedy bandit. Other related works optimize feature engineering using UCB bandits~\cite{anderson2016input} and LLM-based selection queries using online learning~\cite{dai2024uqe}. Another tunes queries with partially obscured UDFs dynamically using runtime statistics~\cite{sikdar2020monsoon}.

\textbf{Other UDF Optimization Methods} Data engines such as SparkSQL treat UDFs as black boxes, relying on distributed computing techniques such as MapReduce to reduce latency~\cite{dean2008mapreduce}.
Our method can be combined with MapReduce by running the indexing and bandit algorithm on each worker, and periodically communicating the running solution back to a coordinator. We do not include such method as baseline, as we assume a single-machine setting for the experiments throughout.

Another line of work studied how to use compilers to optimize Python UDFs within SQL queries~\cite{hagedorn2021putting,grulich2021babelfish,spiegelberg2021tuplex}. 
Python UDF queries tend to be prohibitively slow, as Python has dynamic typing~\cite{spiegelberg2021tuplex} and there are various impedance mismatches between Python and SQL~\cite{grulich2021babelfish}.
Common techniques include compiling Python to SQL~\cite{hagedorn2021putting}, or using a unified IR for Python and SQL to holistically optimize the query and compile into machine code~\cite{grulich2021babelfish,palkar2018evaluating}. 
In our standalone system, we do not have context switching costs, so these techniques are not applicable. 
A full system implementation should adopt these optimization opportunities by compiling each batch of the batched execution model as efficient machine code. 
There are also specialized compilation techniques for PL/SQL UDFs~\cite{hirn2021one}, though they are not as relevant to our target data science workload.

\textbf{Approximate top-$k$ queries} There is a wealth of literature on approximate top-$k$ query answering for various settings. Prior works cover disjunctive queries~\cite{ding2011faster}, distributed networks~\cite{dedzoe2011efficient,cao2004efficient}, and knowledge graphs~\cite{yang2016fast,wang2020semantic,wagner2014pay}. In the classical TA setting~\cite{ilyas2008survey}, budget-constrained~\cite{shmueli2009best}, probabilistic threshold~\cite{theobald2004top}, and anytime~\cite{arai2009anytime} algorithms are known. Ranking problems have also been studied under budget constraints~\cite{politz2011learning,politz2012ranking}. 

\textbf{(DR-)Submodular budget allocation} Our algorithm design and theoretical analysis draw heavily from submodular maximization~\cite{wolsey1982analysis}, especially DR-submodularity over the integer lattice~\cite{soma2018maximizing} and stochastic, monotone submodular maximization~\cite{hellerstein2015discrete,asadpour2016maximizing}. 
Submodular analysis has seen a number of applications, mainly in machine learning~\cite{bilmes2022submodularity}. 
The fixed budget version of our problem generalizes the submodular portfolio problem in \cite{chade2006simultaneous}. 

\textbf{Bandits with nonlinear reward} There are various bandit variants where the total reward is not merely the sum of each iteration's rewards. Prior work studied bandits with convex reward and established a $O(\sqrt{T})$ regret bound~\cite{agrawal2014bandits}. 
Another work studied a bandit problem with DR-submodular reward in data integration and established sublinear regret~\cite{chang2024data}. 
Several works in assortment selection studied bandits with subadditive reward~\cite{goyal2023mnl}. 

\section{Discussion} \label{sec:discussion}

\subsection{Further applications} We focus on the opaque top-$k$ setting, as model-based UDFs are increasingly common. Since the analysis for the top-$k$ bandit is generic, our algorithm has wider applicability. For example, it  can be applied over classic database indexes such as B-trees.  Another potential application is high-priority data acquisition over a union of heterogeneous data sources for model improvement~\cite{chen2023data,chai2022selective}. The scoring function could be proximity to decision boundary, data difficulty, etc. 

\subsection{Fixed budget} We assume that the query execution could terminate at any point.

Hence, we design an anytime algorithm and report the performance of our method at each time step. 
However, prior work shows that knowing the total budget improves solution quality in other top-$k$ settings~\cite{shmueli2009best}. 
In our setting, a budget-constrained algorithm could be risky (i.e., prioritize arms with high variance) earlier and be risk-averse later, and atch all exploration at the beginning.

Computing the budget score ($\text{BS}$ function defined in \S~\ref{sec:algorithm-monotone-dr}) with known budget requires solving an expensive combinatorial problem involving the expectation of order statistics. 
Consequently, a first-principles approach incurs too much overhead. 
Our algorithm skips this computation by being adaptive greedy, computing only the next iteration's marginal gains w.r.t. a realization of $S$. 
A practical alternative is to utilize a variant of Algorithm~\ref{alg:cont-epsgreedy}, batching all exploration rounds at the beginning. 
The number of exploration rounds should be in the order of $\Theta(T^{2/3})$. 

\subsection{Improving the index} We assume that the hierarchical index has been pre-built and is immutable. 
Then, we provide a theoretical performance bound for a given index. 
In practice, scoring functions are likely not completely opaque, so prior knowledge about future queries could help build a better index. 
For example, if queries over some images will be semantic in nature, then light-weight representation learning models could be useful, though it will make the index construction time longer. 
If the scoring functions will primarily use visual information, then using pixel values as we did in this paper would be more effective. 

We may also modify the $k$-means clustering or agglomerative clustering steps. 
Constrained $k$-means clustering might be useful if some leaf clusters are too small to sufficiently explore and exploit~\cite{bradley2000constrained}. 
Since HAC with average linkage is expensive, other linkage types could be more efficient if there are many leaf clusters.

\subsection{In-DBMS implementation}

While we focus on the query execution algorithm, future work could explore how to best incorporate our technique into a practical system. 
A minimal implementation is natural in a system that supports UDFs and an incrementally updating query interface.

We also envision a system specialized for interactive analysis of big data using opaque UDFs. 
It may have a declarative interface for training and prompting models in an ad-hoc manner. 
Such a system should combine data parallelism, statistical optimization, ML systems techniques, and compilers techniques. 
An any-time query model can be used by default for exploratory queries. 
Hand-crafted active learning methods can be used for common types of queries, such as selection queries~\cite{he2020method,dai2024uqe}, aggregation queries~\cite{dai2024uqe}, and top-$k$ queries. 
Other AQP techniques such as non-adaptive sampling algorithms can be used for more complex and general relational queries~\cite{li2018approximate}. 
Then, the entire query plan should be compiled to efficient machine code. 
ML systems techniques such as tuning the model scale for the available hardware~\cite{tan2019efficientnet}, and compressing very large models~\cite{frantar2023qmoe} could be employed to reduce model inference costs. 
How to best tune the query optimizer and execution engine for such a system are problems we leave to future work.

\section{Conclusion} \label{sec:conclusion}

We present an approximate algorithm for opaque top-$k$ query evaluation. Our framework first constructs a generic index over the search domain (\S~\ref{sec:index}). 
Then, we use a novel $\varepsilon$-greedy bandit algorithm for query execution (\S~\ref{sec:algorithms}). We prove that in discrete domains, it approaches a constant-approximation of the optimal (Theorem~\ref{thm:regret}). For practical application, we describe a histogram maintenance strategy that is robust to unknown domains, achieves increasingly higher precision over time, has low overhead (\S~\ref{sec:algorithms}), and can handle failure modes (\S~\ref{sec:fallback}). 
Extensive experiments (\S~\ref{sec:experiments}) of the full framework demonstrate the generality and scalability of our approach to large real-world datasets and a variety of scoring functions. 

\section{Acknowledgement}

This work was supported by the National Science Foundation grant \#2107050. 

\bibliographystyle{ACM-Reference-Format}
\bibliography{main}


\end{document}